\newif\ifnotes

\documentclass[conference,10pt]{IEEEtran}

\usepackage[utf8]{inputenc}

\usepackage[
    backend=biber,
    style=ieee,
    style=numeric-comp,
    sorting=none, 
    isbn=false,
    url=true,
    doi=false,
    url=false,
    mincrossrefs=100,
    maxnames=12, 
    giveninits=true
]{biblatex}

\AtEveryBibitem{\clearname{editor}}
\AtEveryBibitem{\clearfield{note}}

\usepackage{multicol}
\usepackage{algpseudocode}
\usepackage[normalem]{ulem}
\usepackage{setspace}
\usepackage[dvipsnames, table]{xcolor}
\usepackage{soul}
\ifnotes
\usepackage[nomargin, inline, draft]{fixme}
\else
\usepackage[nomargin, inline, final]{fixme}
\fi

\usepackage[anythingbreaks]{breakurl}

\usepackage[T1]{fontenc}
\usepackage{microtype}
\usepackage{pstricks}
\usepackage{graphicx}
\usepackage{amsmath}
\usepackage{xspace}
\usepackage{listings}
\usepackage{multirow} 
\usepackage[most]{tcolorbox}
\usepackage{booktabs}
\usepackage{tabularx}
\usepackage{xspace}
\usepackage{multirow}
\usepackage{pifont}
\usepackage{amsthm}

\usepackage{mathtools}

\usepackage{relsize}
\usepackage{flushend}

\usepackage{caption}
\usepackage{subcaption}
\usepackage{mdframed}

\usepackage{tikz}
\usetikzlibrary{arrows.meta, positioning, calc}

\definecolor{highlight1}{HTML}{F58F33} 
\definecolor{highlight2}{HTML}{2F7BA3} 

\colorlet{highlight1_transparent}{highlight1!50} 
\colorlet{highlight2_transparent}{highlight2!50} 

\newtheorem{definition}{Definition}

\newcommand{\commentout}[1]{}

\definecolor{white}{RGB}{255,255,255}
\definecolor{vrpink}{RGB}{255,0,127}
\definecolor{vrblue}{RGB}{30,144,255}
\definecolor{vrolive}{RGB}{85,107,47}
\definecolor{vrroyalblue}{RGB}{65,105,225}
\definecolor{brgreen}{RGB}{100,200,70}
\definecolor{ivsalmon}{RGB}{255,160,122}
\definecolor{vrlpink}{RGB}{255,192,203}
\definecolor{mvcol}{RGB}{5,150,25}

\usepackage[obeyFinal]{todonotes}

\FXRegisterAuthor{vv}{avr}{\color{vrpink}[vince]}
\FXRegisterAuthor{rvv}{arvr}{\color{red}[vince]}

\newcommand{\gbof}{$\gamma$-batch-order-fairness\xspace}

\usepackage{xcolor}
\usepackage{bbding}
\definecolor{ForestGreen}{rgb}{0.133, 0.543, 0.133}
\newcommand{\checkno}{\textcolor{red}{\XSolidBrush}}
\newcommand{\checkyes}{\textcolor{ForestGreen}{\Checkmark}}

\fxusetheme{color}
\fxuseenvlayout{color}

\usepackage{graphicx}
\graphicspath{{./figures/}}
\DeclareGraphicsExtensions{.pdf,.jpeg,.png,.jpg}

\usepackage{comment}

\usepackage{algorithm}
\usepackage{algpseudocode}
\usepackage{amsmath,amssymb}

\algblockdefx{As}{EndAs}[1]%
{\textbf{as a }#1}%
{\textbf{}}
\algblockdefx{All}{EndAll}[1]%
{\textbf{all replicas}#1}%
{\textbf{}}
\algblockdefx{Timeout}{EndTimeout}[1]%
{\textbf{upon timeout}#1}%
{\textbf{}}
\algblockdefx{Func}{EndFunc}[1]%
{\textbf{function }#1}%
{\textbf{}}
\algblockdefx{WFunc}{EndWFunc}[3]%
{\textbf{function }#1\textbf{ where }#2\textbf{ is }#3}%
{\textbf{}}
\algblockdefx{WFuncx}{EndWFuncx}[3]%
{
  \begin{tabular}{r@{\hspace{0.03in}}l}\textbf{function } & #1\\\textbf{ where } & #2\textbf{ is }#3
\end{tabular}}%
{\textbf{}}
\algblockdefx{WFuncb}{EndWFuncb}[3]%
{\textbf{function }#1\textbf{ where }#2\textbf{ is }#3}%
{\textbf{}}
\algtext*{EndWFuncb}%
\algblockdefx{WWFunc}{EndWWFunc}[5]%
{
  \begin{tabular}{l@{\hspace{0.03in}}r@{\hspace{0.03in}}l}\textbf{function }#1&\textbf{where}&#2\textbf{ is }#3\\&\textbf{and}&#4\textbf{ is }#5
\end{tabular}}%
{\textbf{}}
\algblockdefx{WWWFunc}{EndWWWFunc}[7]%
{
  \begin{tabular}{l@{\hspace{0.03in}}r@{\hspace{0.03in}}l}\textbf{function }#1&\textbf{where}&#2\textbf{ is }#3\\&\textbf{and}&#4\textbf{ is }#5\\&\textbf{and}&#6\textbf{ is }#7
\end{tabular}}%
{\textbf{}}
\algblockdefx{Funcs}{EndFuncs}[1]%
{\textbf{functions }#1}%
{\textbf{}}
\algblockdefx{Funcsb}{EndFuncsb}[1]%
{\textbf{functions }#1}%
{\textbf{}}
\algtext*{EndFuncsb}%
\algblockdefx{Funcb}{EndFuncb}[1]%
{\textbf{function }#1}%
{\textbf{}}
\algtext*{EndFuncb}%
\algnewcommand\lFor[2]{\State\textbf{for}\ #1\ \textbf{do}\ #2}
\algnewcommand\lIf[2]{\State\textbf{if}\ #1\ \textbf{then}\ #2}

\algnewcommand\algorithmicforeach{\textbf{for each}}
\algdef{S}[FOR]{ForEach}[1]{\algorithmicforeach\ #1\ \algorithmicdo}

\newcommand{\hide}[1]{}

\newcommand{\proc}[1]{\textnormal{\textsc{#1}}}
\algnewcommand\Send[1]{\State \textbf{send}\ #1}
\algnewcommand\Broadcast[1]{\State \textbf{broadcast}\ #1}
\algnewcommand\WaitUntil[1]{\State \textbf{wait until}\ #1}

\newcommand{\lc}{\mathit{lc}}
\newcommand{\phiclock}{\varphi}
\newcommand{\tx}{\mathit{tx}}

\usepackage{hyperref}
\hypersetup{colorlinks=true,citecolor=purple}

\makeatletter
\renewenvironment{proof}[1][\proofname]{\par
  \pushQED{\qed}%
  \normalfont \topsep0pt \partopsep0pt 
  \trivlist
\item[\hskip\labelsep \itshape #1\@addpunct{.}]\ignorespaces }{%
  \endtrivlist\@endpefalse
}
\makeatother

\usepackage{mathtools}
\newtagform{blue}{\color{BlueViolet}
(}{)}

\newtheorem{theorem}{Theorem}

\newtheorem{lemma}[theorem]{Lemma}

\newcommand{\name}{Tilikum\xspace}
\newcommand{\nameol}{Tilikum-OL\xspace}
\newcommand{\namebof}{Tilikum-BOF\xspace}
\newcommand{\pompe}{Pomp\=e\xspace}

\begin{document}
\def\sectionautorefname{Sec.} \def\subsectionautorefname{Sec.} \def\figureautorefname{Fig.}
\def\tableautorefname{Tab.} \def\algorithmautorefname{Alg.} \def\equationautorefname{Eq.}

\IEEEoverridecommandlockouts

\title{\name{}: Transaction Fair Ordering on a DAG without Weak Edges}


\author{\IEEEauthorblockN{Giulio Segalini\textsuperscript{*}}
	\IEEEauthorblockA{Université de Neuchâtel\\
		giulio.segalini@unine.ch}
	\and
	\IEEEauthorblockN{Yigit Çolakoğlu}
	\IEEEauthorblockA{Delft University of Technology\\
		Y.Colakoglu@student.tudelft.nl}
	\and
	\IEEEauthorblockN{Marko Putnik}
	\IEEEauthorblockA{Delft University of Technology\\
		M.Putnik@student.tudelft.nl}
    \and
	\IEEEauthorblockN{Jérémie Decouchant}
	\IEEEauthorblockA{Delft University of Technology\\
		j.decouchant@tudelft.nl}
\thanks{\textsuperscript{*}Work performed while at TU Delft.}
}


\makeatletter\def\@IEEEpubidpullup{6.5\baselineskip}\makeatother

\maketitle

\begin{abstract}
  Decentralized Finance (DeFi) applications rely heavily on the order in which transactions are executed, making them susceptible to reordering attacks that enable adversaries to extract Blockchain Extractable Value (BEV).
  While linear blockchain systems such as Ethereum have inspired extensive research into fair ordering mechanisms, DAG-based consensus protocols have remained largely unprotected despite their growing adoption for scalability and performance.
  In this paper, we introduce \name{}, a DAG-based ledger protocol that ensures fair transaction ordering without relying on weak edges.
  \name{} achieves ordering linearizability by leveraging median-based timestamp aggregation, or batch order fairness, while maintaining low data redundancy and robust garbage collection.
  We implemented \name{} in Rust and evaluated it against representative baselines, namely Narwhal/Tusk, \pompe{}, Themis and FairDAG.
  Our results show that \name{} achieves up to $39\times$ higher throughput than other fair-ordering baselines, while fully blocking state-of-the-art DAG-specific reordering attacks.
\end{abstract}


\IEEEpeerreviewmaketitle

\section{Introduction}

Since Bitcoin~\cite{nakamoto2008bitcoin},
cryptocurrencies and blockchain technology have remained a continuous focus of academic research and industrial development. 
Ethereum~\cite{wood2014ethereum} introduced smart
contracts, which allow generalized applications to run on top of blockchains,
achieving strong security guarantees, such as integrity, transparency,
and decentralization. One particular application of smart contracts
that has gained traction is Decentralized Finance (DeFi) systems that offer traditional
financial services without the need for trusted intermediaries~\cite{schueffel2021defi}.

Transactions executed through a DeFi smart contract are often order-dependent,
as the state and outcome of a transaction are dictated by those that precede it. This dependency enables transaction reordering
attacks, where actors extract value by manipulating the order in which transactions are committed or by inserting their own at strategic positions. Because block producers (i.e., miners or validators) have the final authority over block content and ordering, this extracted profit is known as Blockchain Extractable Value (BEV).\footnote{BEV is a more general term than MEV, which specifically refers to blockchains utilizing Proof-of-Work.}

Simple transactions and exchanges between two parties are not susceptible to
these attacks, but other applications are. Some major examples are lending protocols,
which dynamically adjust interest rates depending on previous interactions, and
decentralized exchanges, where buyers and sellers of different tokens are
automatically matched based on previous market interactions. EigenPhi\footnote{\url{https://eigenphi.io/}} reports more than \$$2$M extracted from Ethereum DeFi in a single week (snapshot, June 2024), and Flashbots' MEV-Explore platform~\cite{FlashBot-mevexplore} measured monthly extraction rates around \$$100$M throughout 2022.
These numbers result in an opaque tax paid by users directly or indirectly
through higher fees or worse exchange rates~\cite{daian2019flash}.

BEV mitigation solutions typically follow one of two philosophies. They might first democratize the extraction process by allowing everyone to reap profits from
positioning their transactions favorably. A second approach  enforces fair transaction ordering guarantees. Fair ordering involves computing what would be a fair way to order transactions based on the order in which different replicas receive them~\cite{kelkar2020order,zhang2020byzantine}. These systems define order fairness either using absolute time
indicators (i.e., timestamps) or relative pair-wise orderings to derive a final sequence from replica proposals.

Transaction fair-ordering mechanisms have been mainly designed for linear blockchain structures, where blocks of transactions are committed in a single chain, produced by algorithms such as PBFT~\cite{pbft} or HotStuff~\cite{yin2019hotstuff}. Linear blockchains are also the main support for smart contract platforms like Ethereum~\cite{wood2014ethereum}. However, there is a growing shift toward alternative architectures. In particular, the latest generation of consensus algorithms increasingly utilizes Directed Acyclic Graphs (DAGs) to organize transactions and blocks~\cite{keidar2021all,danezis2022narwhal,spiegelman2022bullshark,jovanovic2025mahi,polyanskii2025starfish}. By decoupling transaction dissemination from ordering, these DAG-based approaches achieve orders-of-magnitude higher throughput without  compromising security.

One might assume that the high performance of DAG-based algorithms would naturally mitigate BEV attacks by narrowing the window of opportunity for attackers.
Zhang et al.~\cite{zhang2025nofish} and Mahe et al.~\cite{mahe2025order} show otherwise: DAG-based ledgers remain susceptible to sophisticated reordering exploits.
Currently, the intersection of DAG-based consensus and fair ordering remains largely unexplored. To the best of our knowledge, the only existing solution is FairDAG~\cite{kang2026fairdag}.
However, FairDAG relies on full transaction redundancy across replicas and on weak edges, both of which preclude garbage collection and limit its practical deployment.
It is also vulnerable to malicious clients that can halt execution by broadcasting transactions to a subset of replicas.
We discuss these limitations in detail in Appx.~\ref{sec:details_FairDAG_tilikum}.

We introduce \name{}, a novel blockchain consensus framework that benefits from the performance of the DAG-based paradigm, supports fair-ordering properties and overcomes FairDAG's limitations. We implement two variants of \name{}, \nameol{} and \namebof{}, that respectively implement ordering linearizability~\cite{zhang2020byzantine} and $\gamma$-batch-order fairness~\cite{kelkar2023themis}. \\

As a summary, this work makes the following \textbf{contributions}. \\
$\bullet$ \textbf{\name{}: ordering linearizability and $\gamma$-batch-order-fairness on a DAG.} We introduce two variants of \name.
\nameol{} establishes fair transaction ordering directly within a DAG by combining a dual-timestamping mechanism with sealed batches: a Logical Table tracks sequence holes, which are synchronized using metadata-rich Hole Fillers, and median-based execution thresholds define safe points for finalizing and executing transactions fairly.
\namebof{} extends this core design to $\gamma$-batch-order-fairness using dependency graphs and digests to resolve ordering cycles, with batch-unspooling and continuity mechanisms to keep overhead low. \\
$\bullet$ \textbf{First weak-edge-free DAG-based fair-ordering protocol.} \name{} is, to the best of our knowledge, the first DAG-based fair-ordering protocol that does not rely on weak edges, thus restoring garbage collection, a prerequisite for production deployment. We formally prove that \nameol{} and \namebof{} guarantee execution safety and liveness, and respectively ensure ordering linearizability and $\gamma$-batch-order fairness. We implement \name{} in Rust as an extension of Narwhal/Tusk~\cite{danezis2022narwhal}, the foundation of subsequent DAG-based algorithms~\cite{spiegelman2022bullshark,babel2025mysticeti,spiegelman2024shoal,arun2025shoal++}; the source code will be open-sourced upon acceptance. \\
$\bullet$ \textbf{Liveness attack on FairDAG with malicious clients.} We identify performance and liveness attacks on FairDAG~\cite{kang2026fairdag} in which malicious clients selectively broadcast transactions to a subset of replicas, harming or halting execution (see Appx.~\ref{sec:details_FairDAG_tilikum}). \name{} is not vulnerable to this attack by design. \\
$\bullet$ \textbf{Evaluation, including under attack.} We benchmark \nameol{} and \namebof{} against state-of-the-art fair-ordering protocols: Themis~\cite{kelkar2023themis}, \pompe{}~\cite{zhang2020byzantine} and FairDAG~\cite{kang2026fairdag}. \nameol{} sustains $14{,}000$\,tx/s at $N{=}10$ with $1.2$\,s latency, $39\times$ Pomp\=e's throughput at the same scale and still $4\times$ faster than Pomp\=e at $N{=}25$. \namebof{} doubles Themis' throughput at every evaluated system size. Under the reordering attacks of Zhang et al.~\cite{zhang2025nofish}, which succeed on vanilla Narwhal/Tusk between $14\%$ and $95\%$ of the time, \nameol{} drops the success rate to $0\%$.

\section{Background on Narwhal/Tusk}

Narwhal/Tusk~\cite{danezis2022narwhal} decouples transaction dissemination from ordering, addressing the bandwidth bottleneck of leader-based BFT. Narwhal serves as the mempool, organizing transactions into a Directed Acyclic Graph (DAG) through a round-based reliable broadcast. Tusk runs on top of Narwhal as the consensus layer.

\textbf{Reliable broadcast and Certificates of Availability.}
When a validator creates a block, it broadcasts it; peers reply with a digital signature once they have validated the block and its $2f+1$ parent certificates. Collecting $2f+1$ signatures forms a Certificate of Availability (CoA), which guarantees that the block's data is stored by at least $f+1$ honest parties and is retrievable even if the original sender fails.

\textbf{Primary-worker architecture.}
For horizontal scalability, each party splits responsibilities between a primary that processes block metadata and several workers that reliably broadcast transaction batches and feed their hashes to the primary.

\textbf{Tusk consensus.}
Tusk groups DAG rounds into three waves: proposal, voting via references, and a shared coin that elects a leader block. A leader is committed once it receives at least $f+1$ references in the next round, and its uncommitted causal history is then ordered deterministically. Clients submit transactions to multiple parties so that, as long as one honest validator receives a transaction, it is eventually included in the total order.

\section{System Model and Objectives}
\label{sec:overview}

\subsection{Blockchain Consensus}

We consider a system of $n$ parties that communicate via message passing. Together, the parties run a Byzantine fault-tolerant state-machine replication (BFT-SMR) protocol that provides a total ordering service for clients. Clients broadcast their transactions to all parties so that they can be ordered as soon as possible to limit the risk of BEV attack. 

We consider a computationally bounded and static adversary that can corrupt up to $f$ parties, which are then said to be faulty or Byzantine. Faulty parties can deviate from the protocol in unrestricted ways and collude, while others are said to be correct.
We assume that clients may also be faulty, e.g., they may send a transaction to a subset of the parties.
Parties have access to a collision-resistant hashing scheme and to an unforgeable digital signature scheme. Each party has a unique private key that allows them to sign messages, and every party knows all public keys and is able to verify all valid signatures.  

We assume that the parties are running an asynchronous DAG-based algorithm, which has been shown to ensure high performance. We build our algorithm, \name{}, on top of Narwhal/Tusk~\cite{danezis2022narwhal} because it has been deployed in production, its code has been publicly released and it has been the basis of several further works~\cite{spiegelman2022bullshark,babel2025mysticeti,polyanskii2025starfish}.

\textbf{Network model.} \name{} is designed for asynchronous but eventually reliable communication links among correct parties, which means that there is no bound on message delays and that an unknown but finite number of messages can be lost. We evaluate \name{} on random delay networks~\cite{byzantine2025danezis}, which are a subcategory of asynchronous networks, by setting parameter $K$ (see \S\ref{sec:parameter_k}) to $f+1$. Supporting any asynchronous network would simply require setting parameter $K$ to $2f+1$, which results in slightly lowered performance (see \S\ref{sec:exp-redundancy}).  

\subsection{Fair-Ordering Objectives}

The location of a transaction in the ledger may impact its successful execution or its outputs.
To protect correct clients against transaction reordering attacks, our goal is to ensure transaction fair-ordering. We aim at supporting (separately) the two state-of-the-art fair ordering properties: ordering linearizability and \gbof.

\begin{definition}[Ordering linearizability~\cite{zhang2020byzantine}]
  \label{def:linearizability} Let ${tx}_1$ and ${tx}_2$ be two client transactions.
  If the highest timestamp provided by a correct party for ${tx}_1$ is
  lower than the smallest timestamp provided by a correct party for
  ${tx}_2$, then all correct parties order (and execute) ${tx}_1$ before ${tx}_2$.
\end{definition}

Ordering linearizability requires committing, for each transaction, $2f+1$ physical-clock timestamps from different parties. Once $2f+1$ such timestamps are committed for a transaction, all correct parties deterministically order it using the median of those timestamps.

\begin{definition}[\gbof~\cite{kelkar2020order,kelkar2023themis}]
  \label{def:gbof} Given a parameter $\frac{1}{2} < \gamma \le 1$, if a fraction $\gamma$ of honest parties receive transaction ${tx}_1$ before transaction ${tx}_2$, then ${tx}_1$ is ordered no later than ${tx}_2$. If these preferences form a cycle, the involved transactions are grouped into a single batch and delivered together.
\end{definition}

\gbof{} therefore requires correct parties to share the order in which they received transactions, and to agree on how cyclic dependencies are resolved.

\textbf{System size.} The number of replicas required by our algorithm, \name{}, depends on the fairness property being enforced. For ordering linearizability, \nameol{} assumes \(n \ge 3f + 1\), which matches the minimum need for consensus. For \(\gamma\)-batch-order-fairness, \namebof{} requires \(n > \frac{f(2\gamma + 1)}{2\gamma - 1}\), which simplifies to \(n \ge 3f + 1\) when \(\gamma = 1\). For comparison, \namebof{} requires the same number of replicas as FairDAG-RL~\cite{kang2026fairdag} to ensure $\gamma$-batch-order-fairness, whereas Themis~\cite{kelkar2023themis} assumes \(n > \frac{f(2\gamma + 2)}{2\gamma - 1}\) (i.e., at least $4f+1$ when \(\gamma=1\))\footnote{As Themis is a leader-based protocol, per-round leaders can only wait to receive $n-f$ local orderings, out of which $f$ might be Byzantine, thus to accommodate the solid edge direction threshold of $2f+1$ (for $\gamma=1.0$) it is required to have $N > 4f$. On the other hand, FairDAG and Tilikum-BOF utilize the validity property of the DAG which guarantees $n - f$ honest local orderings upon commit, thus strengthening the threshold to $N > 3f$.}. 

\subsection{Execution-Ready Performance}

To evaluate the performance of a fair-ordering consensus algorithm, it is necessary to distinguish between consensus and execution performance.
In DAG-based protocols, consensus latency is the time required for a transaction to be committed in the DAG, which is lower-bounded by its first commit at any party. Consensus throughput is the average number of transactions committed per time unit. However, in fair-ordering systems, commitment does not imply immediate execution.
We call the execution latency of a transaction ${tx}_1$ at party $P$ the time from its emission by a client until it is committed in $P$'s DAG and safe to
execute, i.e., when no other transaction ${tx}_2$ can still be ordered before ${tx}_1$.
Execution latency is lower-bounded by the earliest
point in time at which a transaction becomes executable, which can be measured. Execution throughput is the rate at which transactions become executable. 
In fair-ordering systems, execution latency and throughput matter more than consensus ones, with execution throughput always being lower than consensus throughput and execution latency always being higher than consensus latency.

\subsection{Challenges}

Ensuring ordering linearizability and \gbof on a DAG while maintaining high performance is challenging for several reasons.
Unlike in leader-based protocols~\cite{zhang2020byzantine,kelkar2023themis,zarbafian2023lyra}, no single party can be tasked with collecting and committing enough local timestamps or relative orderings.
FairDAG~\cite{kang2026fairdag} has each party commit each transaction in one of its blocks, but since a round only commits $2f{+}1$ DAG blocks and weak edges must be avoided for garbage collection, there is no guarantee that enough information is ever committed per transaction.
Finally, a malicious client may send a transaction to fewer than $2f{+}1$ parties: with too few timestamps to compute a median, execution of later transactions with potentially larger timestamps stalls indefinitely.

\section{\name{}: Core Design for Ordering Linearizability}
\label{sec:details}

This section details the core design of \name, which supports ordering linearizability.

\subsection{Giving Physical Clock Timestamps to Transactions}
\label{sec:parameter_k}

In \name{}, censorship-resistance is guaranteed as soon as a client sends its transaction to at least $f+1$ parties: eventually, the transaction would be communicated to one or several correct parties that would commit it (details in Sec.~\ref{sec:inclusion}).
However, we assume that correct clients send their transactions to all the parties so that each of them gives it a timestamp as early as possible for more robust front-running resistance.
When a party learns about a new transaction $tx$, it verifies whether it is well-formed and correctly signed, and if so it associates it with a timestamp pair, made of the value of its local counter $lc(tx)$ and its UNIX physical clock value $\phi(tx)$.
Afterwards, the party increments its local counter value.
Note that parties give a timestamp to a transaction the first time they see it, which happens either when they receive it from the client or from another replica.
The full per-party state machine is given in Alg.~\ref{alg:timestamp} (Appx.~\ref{sec:pseudocode}). Since every timestamp pair is signed, two conflicting pairs at the same logical counter expose the equivocator and are rejected by correct parties.

\subsection{Including Transactions in Blocks with $2f+1$ Timestamps}
\label{sec:inclusion}

\textbf{Redundancy degree.}
In a DAG, there are two ways to commit transactions and collect $2f{+}1$ local timestamps.

First, all parties may be instructed to include all transactions  in their blocks along with their own timestamp.
Then, the first $2f{+}1$ timestamps associated to a transaction that are committed in the DAG can be used to compute its final timestamp.
This approach requires parties to resubmit a block in a round until it is included in the DAG, or the use of weak edges so that $2f{+}1$ blocks from different parties that contain the transaction are always eventually committed, which limits its applicability to some protocols, such as DAG-Rider~\cite{keidar2021all}.

Second, $K$ parties, where $0 {<} K {\le} n{-}1$, can be tasked with including a subset of the transactions in their blocks along with $2f{+}1$ timestamps that they collect from other parties.
Parties that are tasked with including a transaction in one of their blocks are deterministically selected based on the hash of the transaction.
Depending on the value of $K$, two subcases need to be distinguished. 
If $K \ge f+1$, then at least one correct party will be tasked with including the transaction in the DAG, will collect $2f{+}1$ timestamps from other correct parties and eventually commit a block with the transaction.
If $K \le f$, then it might be possible for a transaction to not be included in the DAG by the original $K$ parties, which might all be faulty. Assuming a correct client, then all parties will eventually learn about the transaction, and can start a per-transaction timer when they receive it. If they do not observe the transaction committed in a block before this timer expires, then they attempt to collect $2f{+}1$ timestamps and commit the transaction in the DAG themselves.
The first approach, where $K \le f$, leads to a lower data redundancy but, in the bad case, its use of timeouts increases execution latency.
\name{} uses $K=f+1$ in random delay networks, and $K=2f{+}1$ in asynchronous networks. 

\medskip
\textbf{Creating batches and blocks.}
Parties that are proposing a batch of transactions need to collect $2f{+}1$ timestamps per transaction.
As in Narwhal, blocks only contain hashes of transaction batches.
For a batch to be included in a block, a party first broadcasts it and waits for $2f{+}1$ signed replies. 
Each such reply contains a vector of timestamp pairs from the sending replica. Signatures on timestamps cover the transaction digest, its timestamp, and the associated logical counter. 
The party then broadcasts its batch with the $2f{+}1$ timestamp vectors it collected for reliable storage, and waits for $2f{+}1$ signed acknowledgments.
A batch is \textbf{sealed} once its proposer has collected $2f{+}1$ timestamp vectors and $2f{+}1$ acknowledgments of reliable storage; only sealed batches can be included in a block.
From here onward, parties operate an instance of the Narwhal~\cite{danezis2022narwhal} protocol normally.
Alg.~\ref{alg:seal} (Appx.~\ref{sec:pseudocode}) formalises the two-round batch-sealing procedure and how a sealed batch's hash is then assembled into a DAG block alongside parent certificates and hole fillers.

\begin{figure}[t]
  \centering
  \includegraphics[width=0.95\linewidth]{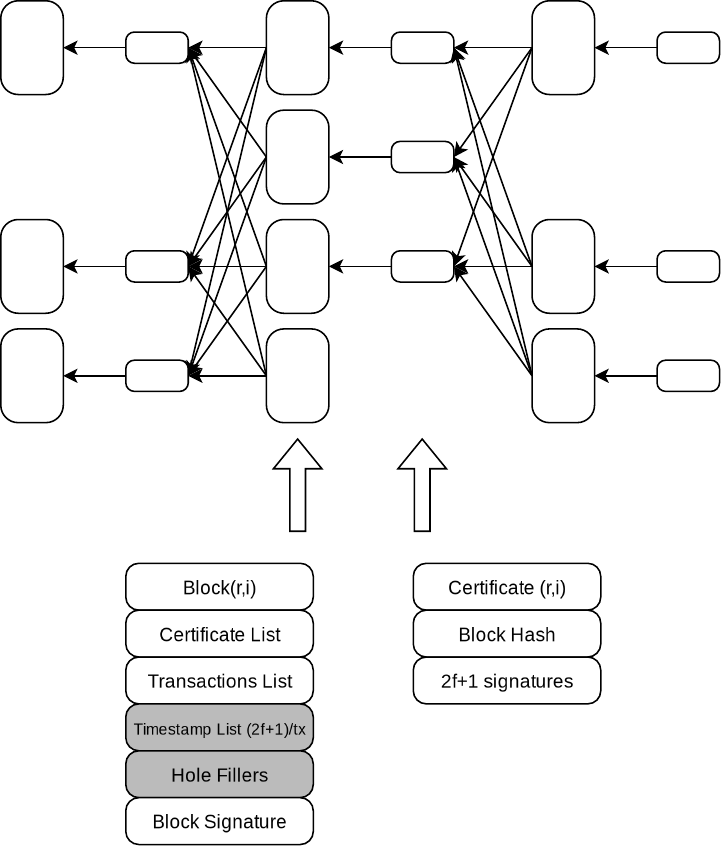}
  \caption{Structure of \name{}'s DAG. Each round contains a set of blocks, represented by large rectangles, while smaller rectangles denote certificates. The DAG excludes weak edges, i.e., references from a block at round $r_b$ to a block at round $r_a < r_b - 1$, enabling efficient garbage collection and practical deployment.}
  \label{fig:dag}
\end{figure}

\subsection{Executing Transactions}

Committed transactions are eventually executed according to the order determined by the median of their $2f{+}1$ timestamps contained in the first block where they appear.
To execute a transaction, a party needs to know for sure that there is no other transaction that has been assigned timestamps by other parties and might be ordered before it.

Each party therefore maintains a list of transactions that are pending execution, which is updated when transactions are executed and when a block of transactions is committed.
Along with these transactions, a party stores their $2f{+}1$ timestamp pairs in the \textbf{Logical Table}: a map from each party's public key to a list of the timestamp pairs emitted by that party, kept sorted by logical-counter value.

The main role of the Logical Table is to track which timestamp pairs are still missing for each other party. A missing pair (an entry whose logical counter has not yet been observed in a committed block) is called a \textbf{hole}. The head of each row is then the earliest timestamp for which the party has seen all predecessors but not the immediate successor.

Parties collect, for each other party, the smallest pair that does not miss any precedent value (i.e., the head of the list), they obtain a list of $3f{+}1$ timestamp pairs that can always be executed. This is because, for each party, parties have also committed all transactions seen by that party with a smaller timestamp than the one in the list.
Parties take the $2f{+}1$ smallest associated UNIX timestamps and compute their median, they then obtain a threshold value until which transaction execution is safe.
This computation is given in Alg.~\ref{alg:threshold} (Appx.~\ref{sec:pseudocode}) and is invoked by Alg.~\ref{alg:exec}.

\begin{algorithm}[t]
\caption{Logical Table Update \& Transaction Execution at party $p_i$}
\label{alg:exec}
\begin{algorithmic}[1]
\State \textbf{Initialization:}
\State $L_i[\,] \gets$ map from party id to sorted list of $(\lc,\,\phiclock)$ pairs
  \Comment{Logical Table}
\State $\mathit{pending}_i \gets \emptyset$
  \Comment{transactions committed but not yet executed}

\Statex
\State \textbf{upon} block $B$ is committed by Tusk \textbf{do}
  \ForAll{sealed batch $sb$ referenced by $B$}
    \ForAll{$\tx \in sb$}
      \State $\mathit{pairs} \gets$ extract $2f{+}1$ timestamp pairs from $sb$
      \State $\tx.\mathit{final\_ts} \gets \proc{Median}\bigl(\{\phiclock \mid (\lc,\,\phiclock) \in \mathit{pairs}\}\bigr)$
      \ForAll{$(p_j,\, \lc_j,\, \phiclock_j) \in \mathit{pairs}$}
        \State $L_i[p_j].\proc{Insert}(\lc_j,\, \phiclock_j)$
      \EndFor
      \State $\mathit{pending}_i \gets \mathit{pending}_i \cup \{\tx\}$
    \EndFor
  \EndFor

  \Statex

  \State{ $\triangleright$ Process hole fillers from $B$'s metadata}
  \ForAll{$(\mathit{author},\, \lc_h,\, \phiclock_h,\, \sigma_h) \in B.\mathit{fillers}$}
    \State \Call{AdvanceRow}{$L_i,\; \mathit{author},\; \lc_h,\; \phiclock_h,\; \sigma_h$}
  \EndFor

  \Statex
  \State{ $\triangleright$ Compute execution threshold and execute}
  \State $\tau \gets \Call{Threshold}{L_i}$
  \State $\mathit{ready} \gets \{\tx \in \mathit{pending}_i \mid \tx.\mathit{final\_ts} \le \tau\}$
  \State sort $\mathit{ready}$ by $\mathit{final\_ts}$ (break ties deterministically)
  \ForAll{$\tx \in \mathit{ready}$ in sorted order}
    \State $\proc{Execute}(\tx)$
    \State $\mathit{pending}_i \gets \mathit{pending}_i \setminus \{\tx\}$
  \EndFor
\end{algorithmic}
\end{algorithm}

\subsection{Hole Fillers}
\label{sec:hole-fillers}

Rows of the Logical Table can stay incomplete. When a transaction is sealed with $2f{+}1$ timestamps, the remaining parties' timestamps never reach the row, leaving holes. Without weak edges, the blocks carrying those timestamps are not guaranteed to be committed either: a certificate of availability may arrive too late to be referenced in the next round and is dropped. Danezis et al.~\cite{danezis2022narwhal} address this by re-injecting batches of uncommitted blocks so that every batch is eventually committed in expectation, but the missing Logical Table entries remain. Since the execution threshold is the median of row heads, a single stuck row halts every later transaction.

To prevent this, we introduce \textbf{Hole Fillers}: metadata added to blocks that allow other parties to catch up on the missing values and advance their Logical Table.
After a block is committed, parties verify for which transactions in the blocks their assigned Timestamp Pair was \textbf{not} included.
All these values should be retrieved from the local storage and included in their own row in the Logical Table.

Parties then add as metadata in the next block the current head of the just updated row.
When processing a committed block, parties can safely advance the logical table row of a block's author such that the new head is equivalent to the hole filler's value.
In case the row is already ahead enough, the value can be simply ignored.

Fig.~\ref{fig:dag} illustrates \name{}'s DAG and highlights in gray the additional information that \name{} commits compared to Narwhal/Tusk: $2f+1$ timestamp pairs per transaction, and a list of hole fillers.  

\begin{algorithm}[t]
\caption{Hole Filler Computation \& Application}
\label{alg:hf}
\begin{algorithmic}[1]
\Function{ComputeHoleFillers}{\,}
  \State $F \gets [\;]$
  \State $\mathit{head} \gets$ head of $L_i[p_i]$
  \State $\sigma \gets \proc{Sign}\bigl(\mathit{head}.\lc,\;\mathit{head}.\phiclock\bigr)$
  \State $F.\proc{Append}\bigl(p_i,\;\mathit{head}.\lc,\;\mathit{head}.\phiclock,\;\sigma\bigr)$

  \ForAll{$(p_j,\, \lc_j,\, \phiclock_j,\, \sigma_j) \in \mathit{pending\_fillers}_i$}
    \State{ $\triangleright$ Include signed fillers from other parties, collected at block receival}
    \State $F.\proc{Append}\bigl(p_j,\;\lc_j,\;\phiclock_j,\;\sigma_j\bigr)$
  \EndFor
  \State \Return $F$
\EndFunction

\Statex

\Procedure{AdvanceRow}{$L,\;\mathit{author},\;\lc_h,\;\phiclock_h,\;\sigma_h$}
  \State $\mathit{cur} \gets$ head of $L[\mathit{author}]$
  \If{$\lc_h > \mathit{cur}.\lc$}
    \State advance $L[\mathit{author}]$ head to $(\lc_h,\;\phiclock_h)$
      \Comment{safe: signed by author}
  \EndIf
\EndProcedure
\end{algorithmic}
\end{algorithm}

\section{\namebof{}: Extension for batch-order-fairness}

\namebof{} reuses \nameol{}'s dual-timestamping, sealed batches, and Logical Table, and adds a per-commit dependency graph (\S\ref{sec:graph-build}) and enriched Hole-Fillers carrying transaction digests. These yield \gbof via batch-unspooling, as in Themis~\cite{kelkar2023themis}.

\subsection{Inferring Local Orderings}

Each transaction $tx$ is bundled with the local counter $lc(tx)$ assigned to that transaction by $2f+1$ parties.
Once a transaction is committed, this information is used to populate the \textbf{logical table}, which is a mapping from each party to its \textbf{logical vector}.
The logical vector $\mathbf{v}$ is an ordered list such that $\mathbf{v}(i) \in T \cup{\emptyset}$ where $T$ is the set of all transactions and $\emptyset$ represents a missing transaction.
When the protocol first starts, all logical vectors are initialized at $\emptyset$.
We denote the logical vector of party $p$ at round $r$ as $\mathbf{v}_p^r$.

As the logical table is updated each round, \name{} uses it to build the transaction dependency graph, since it details the order in which each party received each transaction.
However, the entire logical table cannot be used directly to construct the dependency graph, as it may contain holes (i.e., $\emptyset$ entries) that need to be accounted for.
Instead, we define the local ordering of a party $p$ at a given round $r$ as
$
\mathcal{L}_p^r = \{ tx \in T \mid \forall i \leq lc(tx): \mathbf{v}_p^r(i) \neq \emptyset \}
$.
This method to infer the local ordering for each party ensures that if a transaction $tx$ does not appear in $\mathcal{L}_p^r$, then it is not possible for $p$ to claim that it received a different transaction $tx' \neq tx$ before $tx$, at a later round. This is similar to the \textbf{FIFO} broadcast primitive used in Aequitas~\cite{kelkar2020order}, and plays a key role in implementing batch unspooling.

\subsection{Modified Hole Fillers}

Because the logical table is the main source of information used to infer local orderings, the holes in the logical table need to be filled fully to ensure the protocol's liveness.
For this, we rely on the hole-fillers introduced in Sec.~\ref{sec:hole-fillers}, with some modifications to their content.
For $\gamma$-batch-order fairness, we need to know the transaction associated with each logical counter so that we can correctly compute ordering preferences between transaction pairs.
Therefore, we modify hole-fillers to include the transaction's digest alongside the logical
counter.

Unlike in the ordering-linearizability case, a party must supply fillers for \emph{all} missing logical counters, not just the head of its logical vector.
This imposes additional overhead and makes it harder for lagging parties to publish new local orderings that contribute to the ordering graph.
However, many missing counters may correspond to transactions that are already finalized, and publishing fillers for them is unnecessary.
To address this, we introduce an optimization that reduces redundant filler publication.

\medskip
\textbf{Hole-Filler Continuity and Optimization.}
First, we adjust the protocol such that the hole-fillers published by a party in two consecutive rounds $r, r+1$ are for strictly increasing logical counters and there are no holes between the largest counter in round $r$ and the smallest logical counter in round $r+1$.
This enables a significant optimization: a party can omit fillers for all logical counters associated with already finalized transactions and only publish fillers starting from the unfinalized transaction onward.
This makes it easier for a party that is lagging behind to catch up.

\medskip
\textbf{Handling Unreachable Blocks.}
The correctness of this optimization is void if a block carrying hole-fillers is not included in the chain.
To ensure safety under such conditions, we introduce a two-part sequence identifier for blocks carrying hole-fillers: $(\textit{retry\_count}, \textit{seqnum})$, where
$\textit{seqnum}$ represents the position of the block in the party's filler publication sequence, while $\textit{retry\_count}$ is incremented whenever a block is not committed.

When processing a block, a party applies the optimization only if no $\textit{seqnum}$ values for that party have been skipped, and the $\textit{retry\_count}$ is at least as large as the latest retry count observed for that party. If either condition fails, the optimization is disabled for that block, and no assumptions can be made about earlier holes.
When a party detects that a block containing its hole-fillers was not committed, it continues publishing fillers starting from the missing sequence number but increments its $\textit{retry\_count}$.
This ensures eventual delivery of all required fillers, while still allowing the optimization to operate when the publication sequence is intact.

\subsection{Building the Dependency Graph}
\label{sec:graph-build}
 
Our graph building and unrolling approach follows that of Themis~\cite{kelkar2023themis}, with modifications inspired by FairDAG-RL~\cite{kang2026fairdag}. Similar to Themis and FairDAG-RL we instantiate a fresh dependency graph upon commit and bound its scope to the transactions admitted at that commit, releasing its memory once the graph is sequenced.

\medskip
\textbf{Ordering Notation.} Given a set of local orderings $\mathcal{L}$, we use $tx \in_k \mathcal{L}$ to denote that transaction $tx$ is present in at least $k$ orderings in $\mathcal{L}$. Notation $tx \prec_{(\mathcal{L},k)} tx'$ indicates that $tx$ appears before $tx'$ in at least $k$ orderings in $\mathcal{L}$. Finally, $\text{Weight}_{\mathcal{L}}(tx, tx')$ represents the maximum value $k$ such that $tx \prec_{(\mathcal{L},k)} tx'$. We define $\mathcal{L}_p(lc)$ as party $p$'s transaction with logical counter $lc$ (or $\emptyset$ if none exists) and $\mathcal{L}_p(lc, lc')$ as all transactions in $\mathcal{L}_p$ with logical counters in $[lc, lc']$.

 \medskip
\textbf{Per-commit graphs.}
Each time a leader vertex $L_r$ is committed and its sub-DAG $A_r$ is delivered to the fairness layer, we instantiate a new graph $\mathcal{G}_r$ (Alg.~\ref{alg:processsubdag}). Transactions are admitted to $\mathcal{G}_r$ when they have first been seen by enough parties, that is when $tx \in_{n(1-\gamma) + f + 1} \mathcal{L}(F)$, and remain assigned to $\mathcal{G}_r$ for the remainder of their lifetime. For each pair of admitted transactions $tx, tx' \in \mathcal{G}_r$ we add the edge $tx \to tx'$ once $\text{Weight}_{\mathcal{L}(F)}(tx, tx') \geq n(1-\gamma) + f + 1$ and the reverse edge does not already exist. Pairs whose weights have not yet crossed threshold in either direction are tracked separately and revisited only when new local orderings affect them, restricting per-round work to pairs whose status could plausibly change.

\medskip
\textbf{Lockstep processing.}
We process local orderings for each party in lockstep. For each party $p$ we maintain a monotonically increasing frontier $F_p$ initialized at $0$, representing the last logical counter we have processed for that party. We define $\mathcal{L}(F) = \{\mathcal{L}_p(0, F_p) \mid F_p \in F\}$ as the set of local orderings at the given frontiers. The frontier of party $p$ advances only when its contiguous prefix reaches the next counter, which is the mechanism by which hole-fillers contribute to count growth. The function $\texttt{Order}(\mathcal{L}(F), \mathcal{G}, \gamma)$ processes transactions in $\mathcal{L}(F)$, admits them to the current graph if their count has crossed the admission threshold, and adds edges between admitted pairs whose weights cross threshold. \texttt{Order} is invoked after every leader commit and proceeds across unfinalized graphs in round order.

\medskip
\textbf{Finalization with live count re-evaluation.}
A graph $\mathcal{G}_r$ is finalizable once it is a tournament. We then condense it via Tarjan's algorithm and topologically sort the SCCs. Themis and FairDAG-RL freeze the solid classification at admission time, but we instead re-evaluate the live count at finalization, so a transaction admitted with a count just above the admission threshold can promote to solid as additional parties' prefixes advance through it via hole-fillers. We locate the last SCC containing a transaction whose live count satisfies $tx \in_{2f+1} \mathcal{L}$ and finalize all SCCs at or before this cutoff. Transactions in SCCs after the cutoff have their graph membership cleared and are re-classified against the next leader commit's graph. This live-count check is necessary to preserve liveness for transactions admitted before enough parties' prefixes reach them.

\section{Correctness Proof}

Throughout this section we assume the system model of \S\ref{sec:overview}: an asynchronous, eventually reliable network of $n \geq 3f+1$ parties of which up to $f$ are Byzantine, with all messages signed under a PKI. We first state the safety and liveness properties of the consensus layer (Tusk), following the original proofs~\cite{danezis2022narwhal}. Referring to these proofs, we then prove that \name{} ensures ordering fairness, execution safety and liveness.

\subsection{Safety and Liveness of Consensus}

\textbf{Safety.} Any two honest validators will commit the same sequence of blocks during the consensus step.

\begin{lemma}
  \label{lemma:path} If an honest validator commits block $b$ as leader in
  wave $i$, then any leader block $b'$ committed by honest validators in
  waves after $i$ has a path to $b$.
\end{lemma}

\begin{proof}
  Honest validators commit block $b$ in wave $i$ only if it has $f{+}1$ parties
  in the wave with paths to $b$. All blocks in the first round of wave
  $i+1$ have $2f{+}1$ paths to previous blocks. By Quorum intersection, at
  least one of the paths of the leader block of the wave starting from
  this round will go to $b$. By induction, we can show that every block in
  every round after wave $i$ has a path to $b$. \qed{}
\end{proof}

\begin{lemma}
  \label{lemma:safety} If $b$ is the leader block of wave $i$ and $b'$ is the
  leader block of wave $i'$, if an honest validator commits $b$ before
  $b'$, then no other honest validator commits $b'$ without committing first
  $b$.
\end{lemma}

\begin{proof}
  If an honest validator committed $b$ before $b'$, then there is no path
  from $b$ to $b'$. Assume by absurd that a different validator committed $b
  '$ before $b$ there would also be no path between $b'$ and $b$. But by
  Lemma~\ref{lemma:path} at least one of these paths must exist, making a contradiction. \qed{}
\end{proof}

The safety property is a direct consequence of Lemma
~\ref{lemma:safety}.

\bigskip
\textbf{Liveness.} The liveness property informally says that something is always happening, so
that the system is always progressing towards new correct events. In Tusk, we
want to show that we expect new blocks to be committed in a finite amount of
time. This liveness property is enforced \textbf{only for the consensus step}
and not for the ordering phase. That concept of liveness is defined and
proven later.

\begin{lemma}
  \label{lemma:enough-blocks} In every wave $w$ there are at least $f{+}1$
  blocks in the first round that can be committed.
\end{lemma}

\begin{proof}
  Let's consider any set $S$ of blocks in the \textbf{second} round of
  wave $w$. There are then at least $(2f{+}1)^{2}$ connections to blocks in
  the first round. There also at most $3f{+}1$ blocks in the first round. For
  each of these blocks to be not allowed for commit, they need to have at
  most $f$ connections to blocks in the next round. Assuming the worst case
  where every block in the first round has exactly $f$ connections to blocks
  in $S$, the remaining number of links is still
  $(2f{+}1)^{2}- f(3f{+}1) = f^{2}+ 3f + 1$. Each block in $S$ has at most
  $2f{+}1$ connections to block in the first round. This means there are
  at least $\frac{f^{2}+ 3f + 1}{2f{+}1 - f}$ blocks with more than
  $f{+}1$ connections to the next one. \qed{}
\end{proof}

\begin{lemma}
  \label{lemma:7-rounds} In expectation, a leader is committed every 7
  rounds in a network with asynchronous adversary.
\end{lemma}

\begin{proof}
  By Lemma~\ref{lemma:enough-blocks}, there are at least $f{+}1$
  committable blocks for each new instance of consensus. As the random coin
  is unpredictable, there is a $\frac{1}{3}$ chance of committing a valid
  block. As consensus is run every $3$ rounds, this gives an expected
  value of $9$ rounds. However, the last round of each consensus wave is
  the same as the first round of the next wave, giving us an expectation of
  $7$. \qed{}
\end{proof}

\subsection{Fairness Properties}

This section first discusses the fairness properties that \nameol guarantees.
These are the ordering property, i.e., ordering linearizability or $\gamma$-batch-order-fairness, and the concepts of execution safety and execution liveness. \\

\textbf{Ordering Linearizability.} As \nameol{} adopts the same ordering property and underlying timestamping mechanism as Pomp\=e~\cite{zhang2020byzantine},
their proofs follow a similar logical structure.

\begin{lemma}
  \label{lemma:median} The median of $2f{+}1$ values, of which up to $f$
  can be malicious, is always both upper- and lower-bound by correct
  values.
\end{lemma}

\begin{proof}
  For a set of $2f{+}1$ values, we can see the median as the values
  sandwiched between $f$ smaller (or equal) values and $f$ higher (or equal)
  values. The $f$ malicious values can be spread in three ways: all in the
  smaller set, all in the higher set or spread between them. In the first
  2 cases, we have that the median itself is a correct value and that the other
  set is fully composed of correct values, bounding the median between
  honest numbers. In the last case, we have that each set contains at
  least one correct element, which is smaller or equal for the left-side and
  higher or equal on the right-side, again bounding the median between
  these two correct entries. \qed{}
\end{proof}

\begin{lemma}
  \label{lemma:linearizability} Transactions ordered by the median of $2f{+}
  1$ collected timestamps respect ordering linearizability (Definition~\ref{def:linearizability}).
\end{lemma}

\begin{proof}
  Let us take two transactions ${tx}_{1}$ and ${tx}_{2}$ with the respective
  median timestamps $t_{1}$ and $t_{2}$. If the maximum timestamp from a
  correct party for $tx_{1}$ is $tm_{1}$ and the minimum timestamp from a
  correct party for $tx_{2}$ is $tm_{2}$, we know that $t_{1}<= tm_{1}$ and
  $t_{2}>= tm_{2}$. Hence, if $tm_{1}< tm_{2}$, the first condition of
  Definition~\ref{def:linearizability}, we know that $t_{1}\le tm_{1}< tm_{2}
  \le t_{2}$ and so $t_{1}< t_{2}$, which will make correct parties order, and
  so execute, $tx_{1}$ before $tx_{2}$. \qed{}
\end{proof}

\bigskip
\textbf{$\gamma$-batch-order fairness.} \namebof{}'s graph layer mirrors that of FairDAG-RL~\cite{kang2026fairdag}, so its $\gamma$-batch-order-fairness theorem transfers once we establish that our input local orderings and our admission rule coincide with theirs.
\begin{lemma}[Local-ordering soundness]
\label{lemma:lo-sound}
For every correct party $p$ and round $r$, the hole-free prefix of $\mathbf{v}_p^r$ equals a prefix of the order in which $p$ first observed transactions, and each $(p, \ell)$ pair binds to exactly one transaction.
\end{lemma}
\begin{proof}
A correct party $p$ runs Alg.~\ref{alg:timestamp}, where $\lc$ is incremented monotonically and a counter is assigned only on first sight, so the mapping $\ell \mapsto tx$ at $p$ is injective and the resulting sequence is by construction $p$'s reception order. Hole fillers carry $p$-signed $(\ell, tx)$ records (Alg.~\ref{alg:hf}), so no party can inject a counter on $p$'s behalf, and any $(p,\ell)$ entry that reaches the logical table is exactly the one $p$ assigned. \qed{}
\end{proof}
\begin{lemma}[Logical-table agreement]
\label{lemma:lo-agree}
At every round $r$, all correct parties observe identical logical tables.
\end{lemma}
\begin{proof}
By Lemma~\ref{lemma:safety}, correct parties commit identical sequences of leader blocks and therefore identical causal sub-DAGs. The logical table is a deterministic function of the committed blocks' sealed batches and hole fillers (Alg.~\ref{alg:exec}), and identical input yields identical state. \qed{}
\end{proof}
\begin{theorem}[\gbof]
\label{theor:bof}
\namebof{} satisfies $\gamma$-batch-order-fairness.
\end{theorem}
\begin{proof}
Suppose at least $\gamma(n-f)$ correct parties received $tx$ before $tx'$. By Lemma~\ref{lemma:lo-sound}, this is reflected in the local orderings $\mathcal{L}_p$ each correct $p$ contributes. By Lemma~\ref{lemma:lo-agree}, every correct party runs the graph builder on the same logical table $\mathcal{L}$.

On this shared input, \namebof{} admits $tx$ to graph $\mathcal{G}_r$ when its count reaches $n(1-\gamma) + f + 1$ and adds the edge $tx \to tx'$ once $\text{Weight}_\mathcal{L}(tx, tx') \geq n(1-\gamma)+f+1$. Both rules and the per-commit graph structure coincide with FairDAG-RL's graph construction approach. Finalisation emits an SCC only once it contains a transaction satisfying $tx \in_{2f+1} \mathcal{L}$, i.e., passed the solid threshold, after which it processes SCCs in topological order of the condensation and defers transactions in later SCCs to subsequent rounds, which is FairDAG-RL's \textsc{OrderFinalization} rule~\cite{kang2026fairdag}.

\namebof{} differs from FairDAG-RL only in that we evaluate the solid predicate against the live $\textit{count}(tx)$ at finalisation rather than against the type frozen at admission. Since count grows monotonically, this admits exactly the FairDAG-RL finalisations plus those where a transaction has reached $\in_{2f+1} \mathcal{L}$ between admission and tournament completion, so any ordering FairDAG-RL would produce on $\mathcal{L}$ is also produced by \namebof{}. FairDAG-RL's Theorem 8.18~\cite{kang2026fairdag} applies, yielding $\gamma$-batch-order-fairness. \qed{}
\end{proof}

\subsection{Safety and Liveness of Execution}

\textbf{Safety.}
We now prove that the threshold computed through the Logical Table is always
safe, which means that it is always lower than the assigned
timestamp of any transaction that may still possibly be committed.

\begin{lemma}
  \label{lemma:threshold-smaller-hole} The assigned timestamp $t$ of a transaction
  $tx$ is not greater than the threshold value that would be calculated if
  $tx$ is the first hole for every party.
\end{lemma}

\begin{proof}
  The assigned timestamp $t$ is the median of the timestamp assigned to $tx$
  by $2f{+}1$ parties. The threshold computed by the Logical Table in
  case this transaction is the smallest hole for all parties is the
  median of the smallest $2f{+}1$ timestamps assigned by parties. This value
  cannot then be bigger than $t$. \qed{}
\end{proof}

\begin{theorem}
  \label{theor:threshold-valid}
  Taking an arbitrary uncommitted transaction $tx$ with assigned timestamp
  $t$, the threshold $tr$ cannot be larger than $t$.
\end{theorem}

\begin{proof}
  An uncommitted transaction will always result
  in a hole for every entry in the Logical Table. According to Lemma~\ref{lemma:threshold-smaller-hole},
  the threshold resulting from the situation where this is the smallest hole
  for each party, results in the threshold being valid for future
  execution of this transaction. In case this transaction is not the first
  hole for some or all parties, we know that the threshold will then be
  committed over smaller holes, which always have smaller or equal timestamps, as parties increase their timestamps monotonically. \qed{}
\end{proof}
Theorem~\ref{theor:threshold-valid} shows that it is impossible for the threshold
to be larger than the timestamp of an uncommitted transaction, showing
Execution Safety.

\textbf{Liveness.}
Remember that transactions
are first committed through a consensus mechanism (Tusk) and execution is delayed
until the \texttt{threshold} value is high enough for safe execution.
The following lemma first shows that on average and with random network latencies/delays, the threshold is always increasing and so some transaction is always eventually executed.

\begin{lemma}
  \label{lemma:prob-block} Assuming random network latencies and that every
  (correct) party proposes one block per round, each proposed block has
  a $\frac{2}{3}$ probability of getting approved in the current round.
\end{lemma}

\begin{proof}
  Having random network delays implies that the order a party receives blocks
  is random. If $3f{+}1$ blocks are submitted each round and each party
  waits for $2f{+}1$ certificates before progressing, then from the point of
  view of each party $\frac{2f{+}1}{3f{+}1}> \frac{2}{3}$ blocks will be
  included in the current round. \qed{}
\end{proof}

\begin{lemma}
  The expected value of committed blocks per party and per wave of consensus
  is 3.
\end{lemma}

\begin{proof}
  From Lemma~\ref{lemma:7-rounds} we know that on average consensus happens
  every 7 rounds. From Lemma~\ref{lemma:prob-block} we get that
  $7\cdot\frac{2}{3}$ blocks per author are in a wave. Moreover, by $2/3$-Causality
  $\frac{2}{3}$ of these blocks will be committed, giving us an expectation
  of $7\cdot(\frac{2}{3})^{2}\approx 3.11$. \qed{}
\end{proof}

\begin{lemma}
  \label{lemma:fp1-updated} Every time a leader block is committed,
  the logical table row of at least $f{+}1$ correct parties is updated.
\end{lemma}

\begin{proof}
  If a leader block is successfully committed, then from $1/2$-Chain
  Quality we know that at least of half the blocks in the causal history are
  from correct authors. These blocks will include timestamps from $2f{+}1$
  unique sources and so at least $f{+}1$ will be correct. \qed{}
\end{proof}

\begin{lemma}
  \label{theor:avg-liveness} In expectation, each round of consensus will
  update the local logical clock of each correct party.
\end{lemma}

\begin{proof}
  As in expectation $3$ blocks per party will be committed and these blocks
  will definitely contain at least $1$ timestamp pair from its author,
  then their entry in the logical table will be updated. \qed{}
\end{proof}

Lemmas~\ref{lemma:fp1-updated} and~\ref{theor:avg-liveness} give us that
on average, the logical table of all correct processes, with a minimum of $f{+}
1$ per wave, will be updated.
These two lemmas together prove the final theorem that guarantees execution
liveness:

\begin{theorem}
  In expectation, the threshold computed locally by each party is always
  increasing, and new transactions are always being executed.
\end{theorem}

\begin{proof}
  By the safety property, each correct party will commit the same
  sequence of blocks and so have the same internal status for the Logical
  Table. The threshold is updated by taking the median value of the $2f{+}1$
  smallest timestamps for which we miss the successive value in the logical
  table. From lemma~\ref{lemma:median}, this value is both lower- and
  upper-bound by correct values and if all correct values are getting
  updated, then the threshold also gets updated. As timestamps cannot
  decrease, the threshold is always increasing. \qed{}
\end{proof}

This guarantees Execution Liveness in the average case, i.e., in random delay networks. We can also show
that it holds in the general case by exploiting the
fact that every block can be resubmitted and will therefore eventually be committed~\cite{danezis2022narwhal}.
As a consequence, since every correct party will try
to include in their blocks every transaction they have been assigned together with hole fillers, then every value in the Logical Table of correct parties will eventually be committed, and the execution threshold will regularly increase.
This second way of looking at Execution Liveness gives us the guarantee that the threshold will eventually catch up, but does not allow us to perform any theoretical analysis on the additional latency.

\begin{theorem}
  \label{theor:liveness-async}
  Under full asynchrony, if $K = 2f{+}1$ (Sec.~\ref{sec:parameter_k}),
  then the threshold computed by each correct party strictly increases
  infinitely often, and every submitted transaction is eventually executed.
\end{theorem}

\begin{proof}
  Each transaction $tx$ is assigned to $2f{+}1$ parties, of which at least
  $f{+}1$ are correct. A round is certified when $2f{+}1$ block authors are
  included, so it excludes the blocks of at most $f$ authors. By pigeonhole,
  at least one of the $f{+}1$ correct parties tasked with $tx$ has its block
  certified in every round following $tx$'s submission, hence $tx$ appears
  in at least one certified block of a correct party.
  By the eventual-commit property of the DAG~\cite{danezis2022narwhal}, this block
  is eventually committed, so the system has no infinite holes due to
  missing transactions.

  Each committed $tx$ is bundled with $2f{+}1$ timestamp pairs, and any
  correct party whose pair is not among them triggers a hole filler
  (Sec.~\ref{sec:hole-fillers}) that is inserted in a subsequent block
  and fills the corresponding entry. By Lemma~\ref{lemma:fp1-updated} every
  committed leader updates the row of at least $f{+}1$ correct parties, and
  by Lemma~\ref{lemma:median} the median of the $2f{+}1$ smallest hole
  timestamps is bounded by correct values. The threshold is therefore
  monotonically increasing.
  By Lemma~\ref{lemma:threshold-smaller-hole} it eventually surpasses $tx$'s
  assigned timestamp, and $tx$ is executed. \qed{}
\end{proof}

\begin{figure*}[t]
  \centering
  
  \begin{minipage}{0.48\textwidth}
    \centering
    \includegraphics[width=\linewidth]{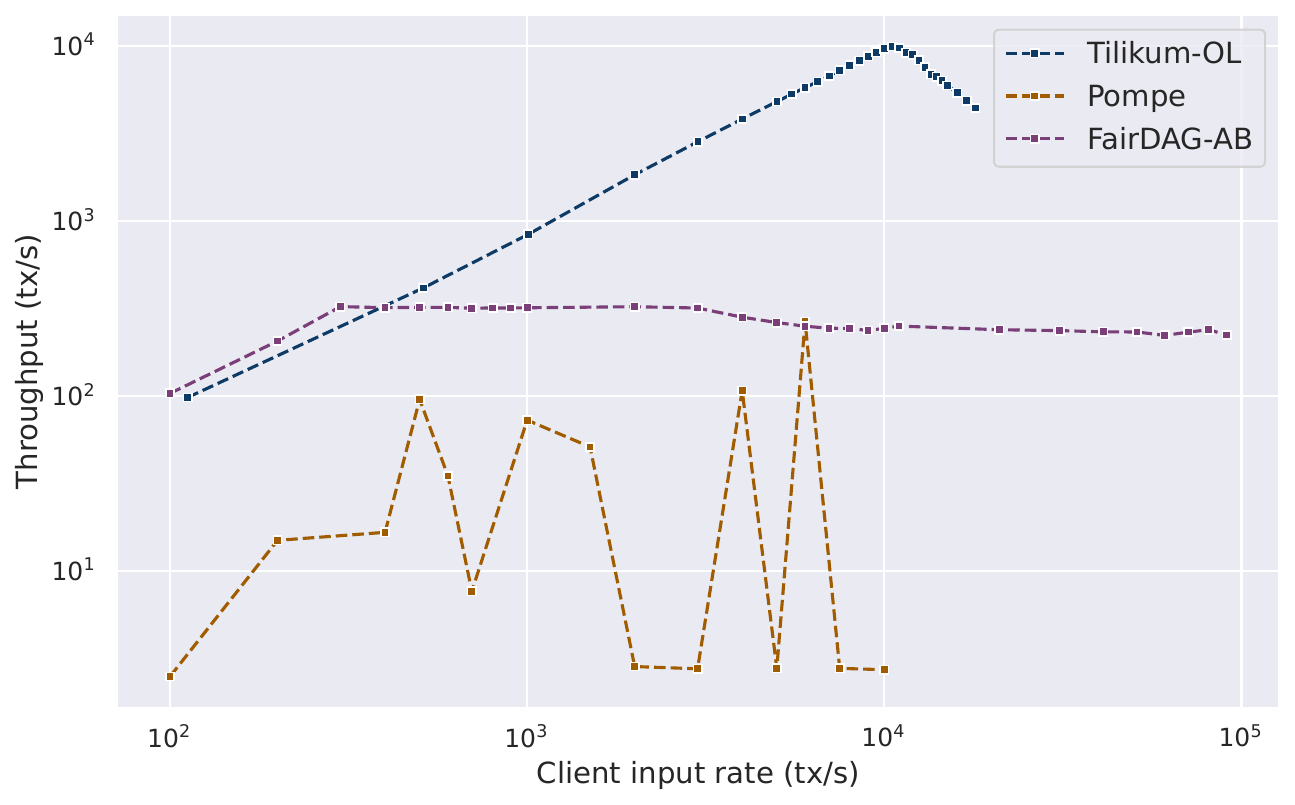}
    \caption{Ordering Linearizability --- Execution throughput ($\uparrow$) depending on transaction input rate ($n=16$)}
    \label{fig:rate-tp-ol}
  \end{minipage}
  \hfill 
  \begin{minipage}{0.48\textwidth}
    \centering
    \includegraphics[width=\linewidth]{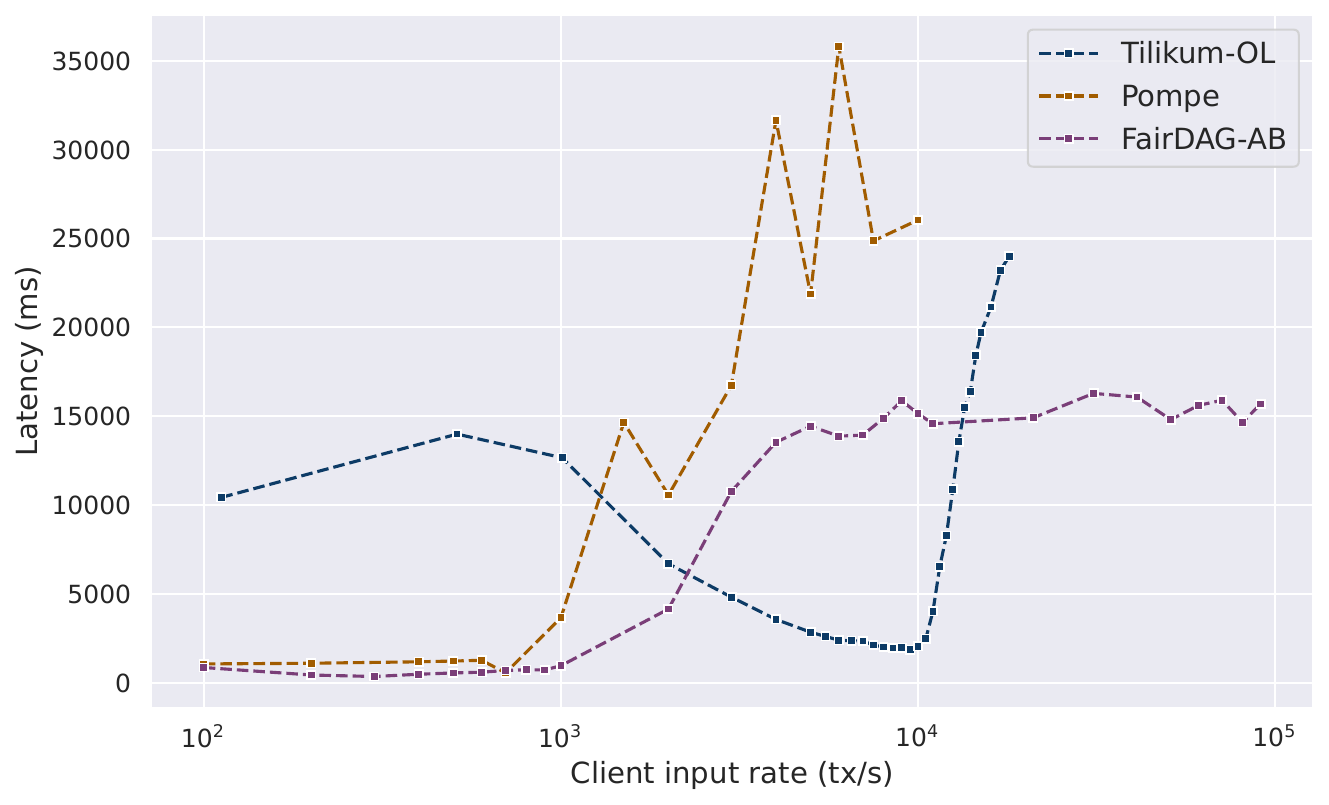}
    \caption{Ordering Linearizability --- Execution latency ($\downarrow$) depending on transaction input rate ($n=16$)}
    \label{fig:rate-lat-ol}
  \end{minipage}

  \vspace{2em} 

  \begin{minipage}{0.48\textwidth}
    \centering
    \includegraphics[width=\linewidth]{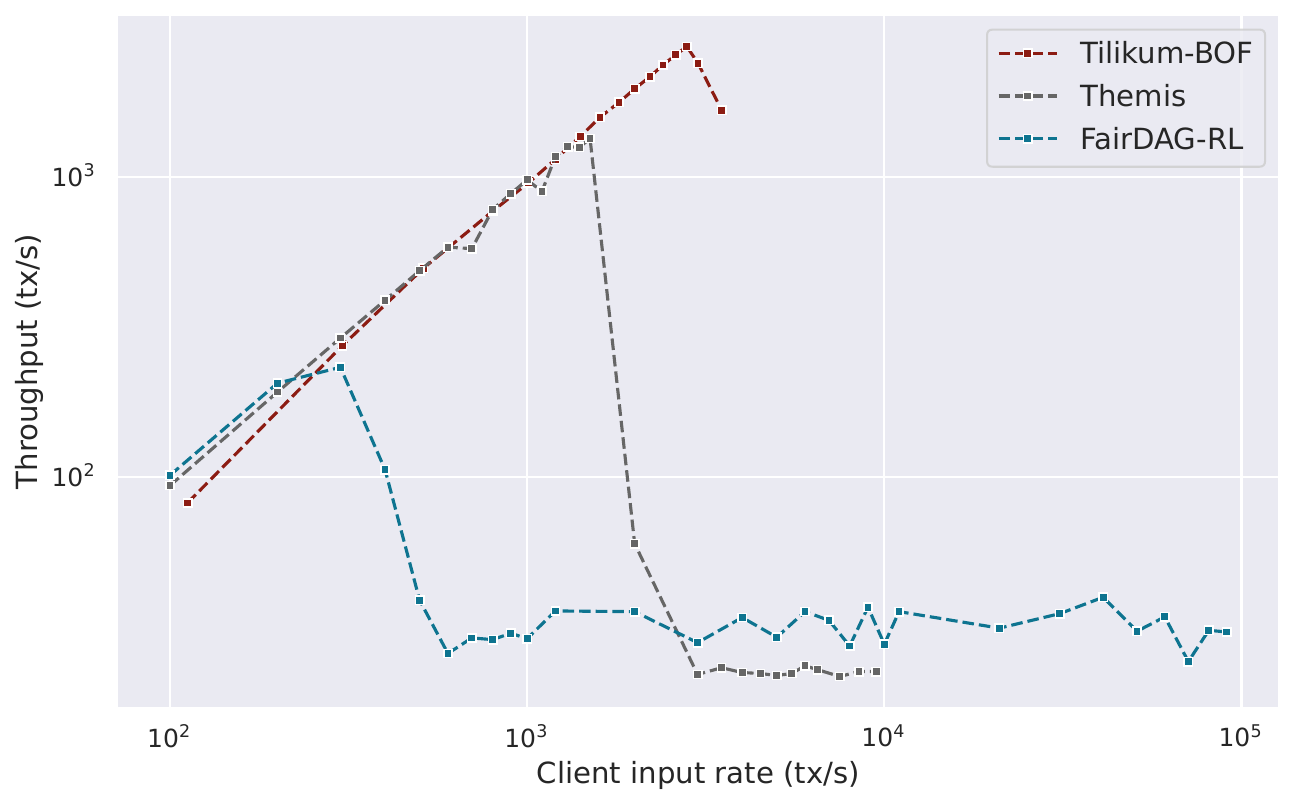}
    \caption{\gbof --- Execution throughput ($\uparrow$) depending on transaction input rate ($n=16$)}
    \label{fig:rate-tp-bof}
  \end{minipage}
  \hfill
  \begin{minipage}{0.48\textwidth}
    \centering
    \includegraphics[width=\linewidth]{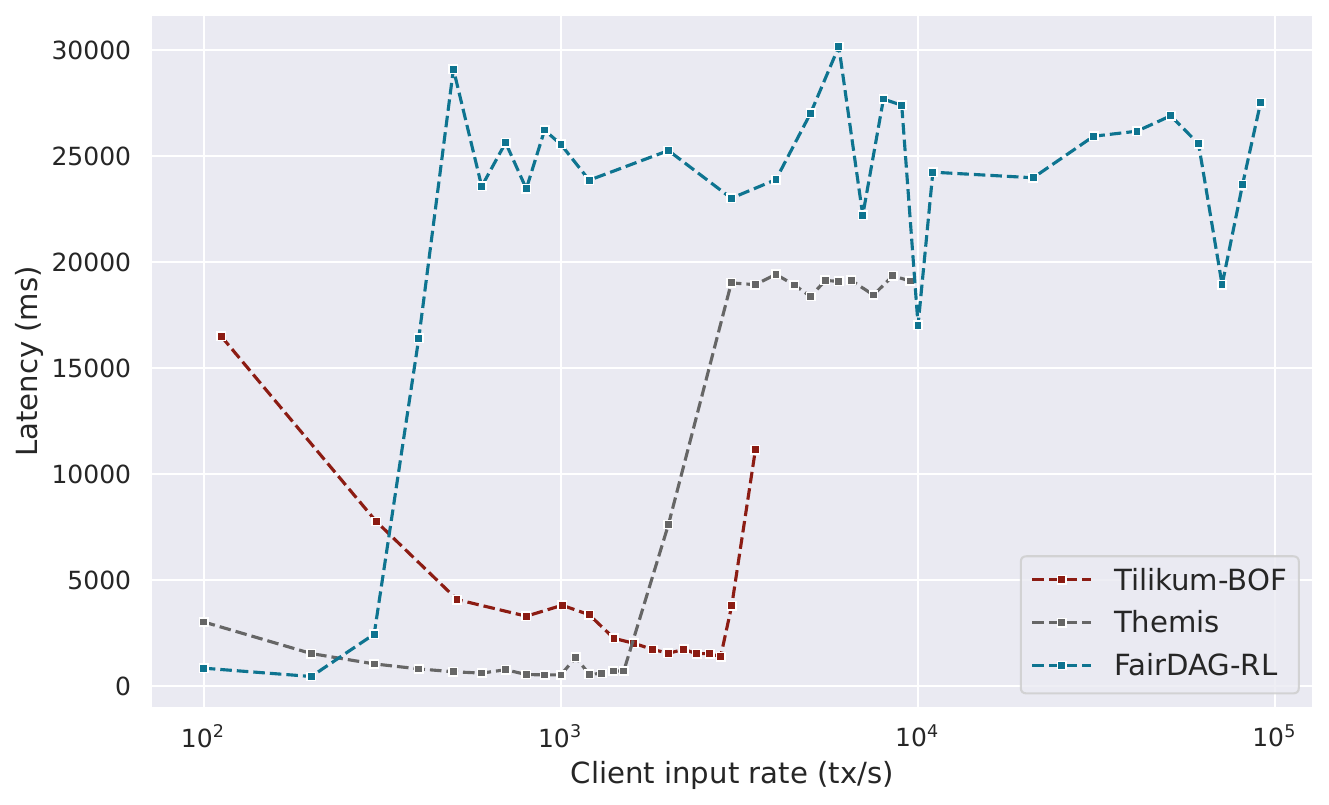}
    \caption{\gbof --- Execution latency ($\downarrow$) depending on transaction input rate ($n=16$)}
    \label{fig:rate-lat-bof}
  \end{minipage}
\end{figure*}

\section{Performance Evaluation}

\subsection{Setup}

We implemented \name starting from Narwhal/Tusk's Rust original implementation\footnote{\url{https://github.com/asonnino/narwhal/}}.
Our modifications span a total of $4,000$ LOCs, while our experimental scripts required $2,500$ additional LOCs. 
We backported some features of the production-ready implementation of Narwhal by Sui\footnote{\url{https://github.com/MystenLabs/sui/blob/narwhal-votes/narwhal}}, namely the re-transmission of uncommitted certificates and the order in which batches are added to a block.
Except when indicated otherwise (i.e., in \S\ref{sec:exp-redundancy}), we deploy the random delay network version of \name in which $K=f+1$ parties are tasked with including a transaction in one of their block. \\

\textbf{Baselines.} We compare \name{}'s throughput and latency against those of Pomp\=e~\cite{zhang2020byzantine}, Themis~\cite{kelkar2023themis}, the FairDAG variants~\cite{kang2026fairdag} and the vanilla Narwhal/Tusk~\cite{danezis2022narwhal}. Pomp\=e implements ordering linearizability on top of HotStuff~\cite{yin2019hotstuff}.  FairDAG-AB and FairDAG-RL respectively support ordering linearizability and \gbof{} on top of any DAG-based algorithms that uses weak edges.

We use FairDAG's code\footnote{\url{https://github.com/apache/incubator-resilientdb/commit/9e6c46f6f1d56ed88aae1034699c12d2097b1313}}, which we slightly adapted for a fair comparison against \name.
Our changes are described at the end of the section together with a theoretical comparison of \name and FairDAG~\cite{kang2026fairdag}, focusing in particular on maximum theoretical throughput, real-life feasibility and design decisions.
FairDAG is built on top of ResilientDB~\cite{gupta2020resilientdb} and has a different concept of batches.
A batch is a collection of transactions from the same client and is assigned a single timestamp.
To fairly compare against \name and Pomp\=e~\cite{zhang2020byzantine}, the batch size was set to $1$.
This value indicates the amount of client-submitted transactions that are treated as a single one in the system, which also mean a single timestamp is assigned to this batch. 

The released implementation of FairDAG does not support the simulation of faulty parties. Consequently, we did not evaluate it under this scenario.
FairDAG's numbers may underrepresent its potential because of ResilientDB overhead; we contacted the authors to confirm our setup. Either way, its reliance on weak edges and high redundancy (Appx.~\ref{sec:details_FairDAG_tilikum}) would have to be removed for production use, which would only lower its throughput further. \\

\textbf{Experimental settings.}
All experiments were run on the DAS5 distributed cluster~\cite{bal2016medium}.
Each individual run lasted $60$\,s and each experiment is repeated $5$ times.
Clients are assigned their own machine, while a party's workers and primary are colocated. 
%
Transactions have a fixed size of $128$ bytes. Pomp\=e's implementation is based on \texttt{libhotstuff} which does not exchange transactions, it instead reaches consensus on their $32$ bytes cryptographic hash digest, assuming that the actual transaction is sent for execution when confirmed.
FairDAG's implementation is based on ResilientDB~\cite{gupta2020resilientdb} and operates analogously.
Having the transaction size to be around the same size of the exchanged hash values in these protocols makes the comparison with \name{} as fair as possible.

Workers build batches of $2,000$ bytes (around $16$ transactions) or after a $1$\,s timeout.
Primaries build blocks containing $256$ bytes of payload, which is equivalent to $8$ batches (each batch is represented by a $32$ bytes hash).
A block can also be built after a timeout of $2$ seconds.

In \name, transactions are broadcast by clients to all parties and when a worker receives a transaction an extra communication step to the primary is required to assign a Timestamp Pair.
We kept the Narwhal/Tusk experiment code in which clients send transactions to only one party so that the maximum achievable throughput can be showcased.

The security and fairness evaluation was performed with fixed parameters, clients sent $13{,}000$ transactions per second, with $10$ clients deployed; this rate sits at \nameol{}'s saturation point in Figs.~\ref{fig:rate-tp-ol} and~\ref{fig:rate-lat-ol}, so the system runs under full load while we measure attack effects. Block and batch sizes were left unchanged.
Two sets of experiments were performed, one with no faulty (silent) parties and varying number of arbitragers from $1$ to $5$ and another with a fixed amount of $3$ arbitragers and varying the number of silent parties from $0$ to $3$.

It is important to note that while the sluggish and speculative attack do not imply that the party deviates from the protocol, the Fissure attack might lead to that happening.
The experiments executing the fissure attack were carried out with the total number of malicious parties (arbitragers $+$ silent) below $1 /3$ of the total.

\subsection{Throughput and Latency Depending on Injection Rate}

We first ran a set of experiments with $N=16$ parties where we aimed at finding the maximum achievable input rate for each algorithm.
We progressively increased the input rate until the effective throughput stopped increasing or a substantial increase in latency was observed.
Latency and throughput are reported in Figs.~\ref{fig:rate-tp-ol}--\ref{fig:rate-lat-bof}. 
Note that with transaction input rate we mean the amount of unique transactions inserted in the system, which results in at least $f+1$ messages for each one.
Most protocols follow a similar pattern: inserting more transactions increases throughput linearly until a stall point is reached.
The system becomes saturated and we see an increase in latency and a decrease or flattening of throughput.
All of \name's versions have higher latency at low input rates.
This is a consequence of the system waiting for enough transactions to fill batches and blocks, so latency is mainly driven by the timeout settings.





\begin{figure}[t]
  \centering
  \includegraphics[width=1\linewidth]{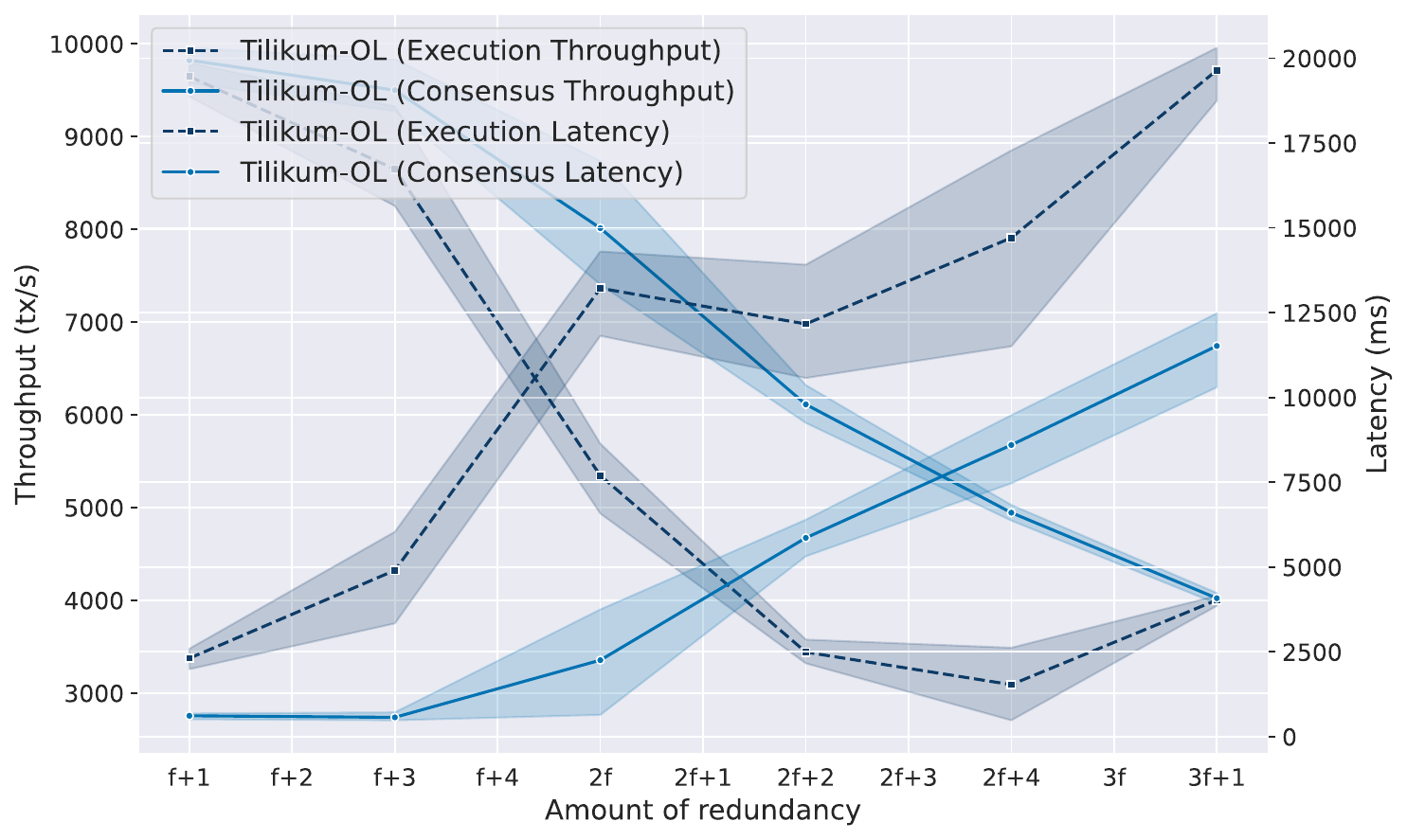}
  \caption{Ordering Linearizability --- Consensus and execution throughput ($\uparrow$, downward trend curves) and latency ($\downarrow$, upward trend curves) as a function of redundancy degree $K$ ($n=16$)}
  \label{fig:redundancy}
\end{figure}

\subsection{Evaluating Data Redundancy}
\label{sec:exp-redundancy}

Fig.~\ref{fig:redundancy} reports the latency, consensus and execution throughput of \name's variants when changing the number $K$ of parties that are assigned to include a transaction.
We run these experiments with $16$ parties, all parties behaving honestly, and varying the number $K$ from $4$ to $10$.
%
Execution throughput starts at approximately $9{,}500$--$10{,}000$~tx/s for redundancy levels between $f+1$ and $f+3$, then drops sharply to about $5{,}300$~tx/s at $2f$, and further decreases to roughly $3{,}000$--$4{,}000$~tx/s near $2f+2$--$3f+1$. In contrast, consensus throughput decreases more gradually, from around $6{,}700$~tx/s at $f+1$ to approximately $4{,}000$~tx/s at $3f+1$.
Latency exhibits the opposite trend. Consensus latency increases steadily from roughly $200$--$400$~ms at low redundancy levels to approximately $13{,}000$--$14{,}000$~ms at $3f+1$. Execution latency grows even more sharply, rising from around $2{,}000$--$3{,}000$~ms at $f+1$ to nearly $20{,}000$~ms at $3f+1$.



\begin{figure}[t]
  \centering
  
  \begin{minipage}[t]{0.48\textwidth}
    \centering
    \includegraphics[width=\linewidth]{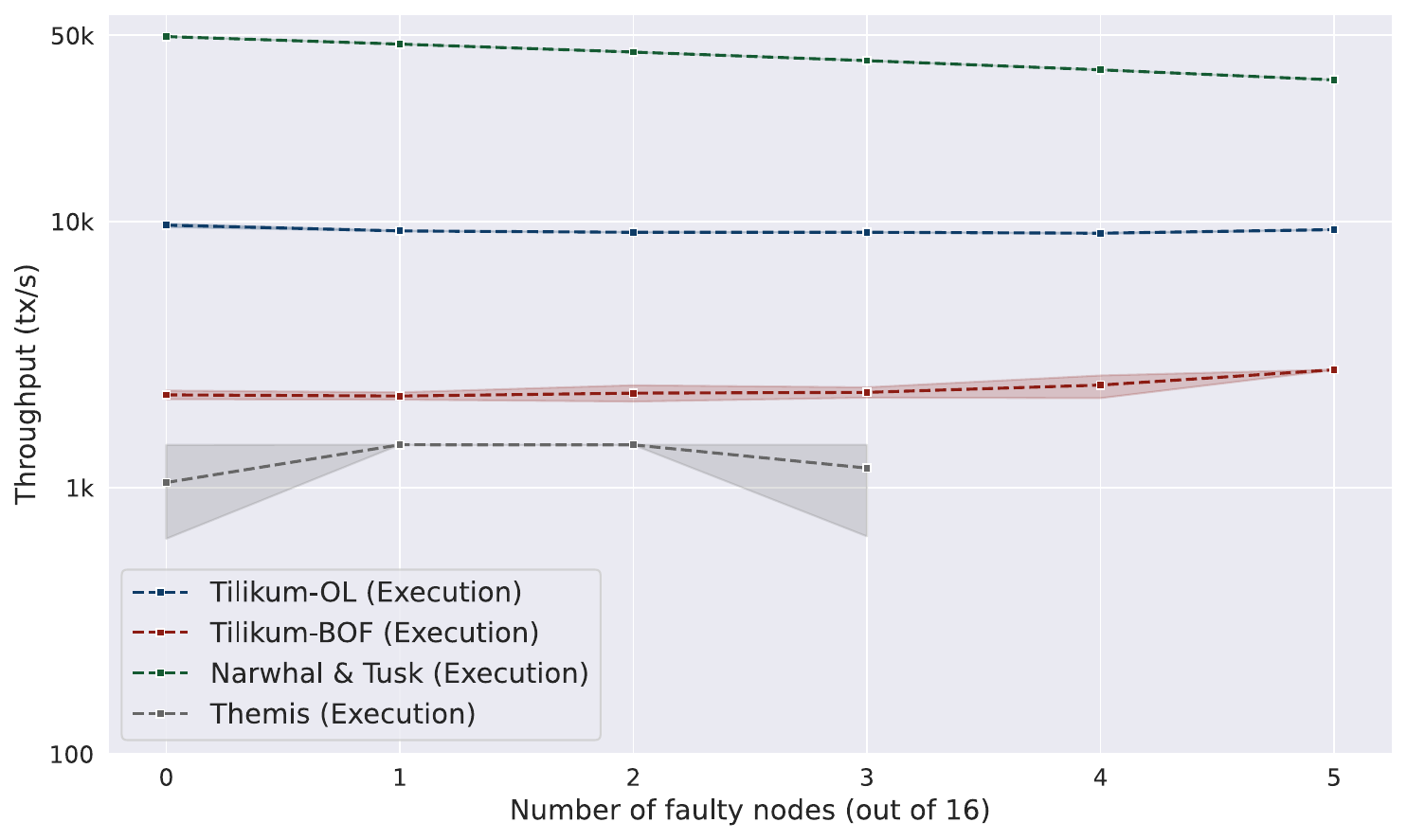}
    \caption{Execution throughput ($\uparrow$) depending on the actual number of faulty parties ($n=16$, with respectively up to $5$ and $3$ faults with OL and BOF)}
    \label{fig:throughput-f}
  \end{minipage}
  \hfill
  \begin{minipage}[t]{0.48\textwidth}
    \centering
    \includegraphics[width=\linewidth]{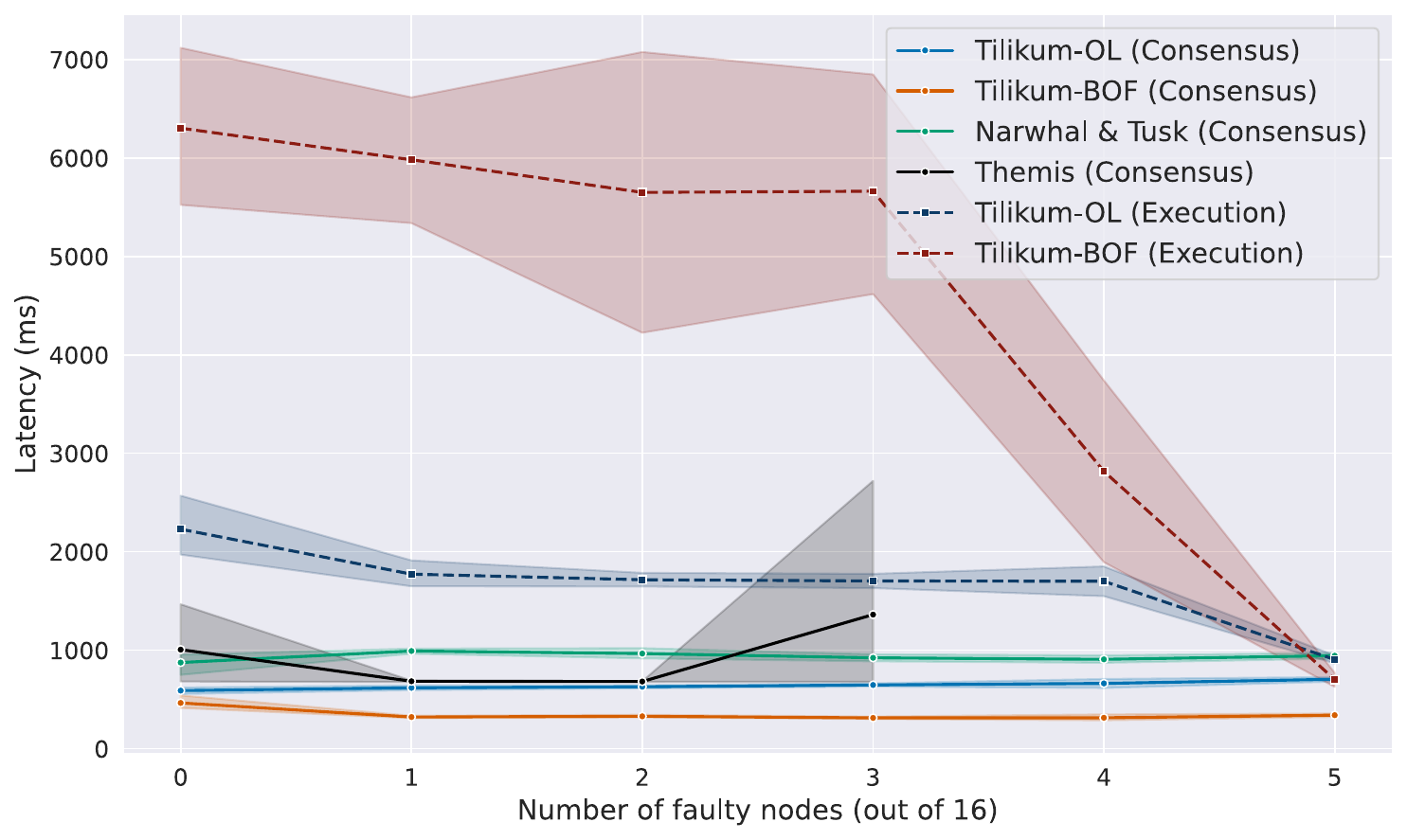}
    \caption{Execution latency ($\downarrow$) depending on the number of faulty parties ($n=16$)}
    \label{fig:latency-f}
  \end{minipage}
\end{figure}

\subsection{Throughput and Latency Depending on System Size}

We now evaluate the throughput and latency of \nameol{} and \namebof{}, benchmarking them against existing state-of-the-art fair ordering protocols. \\

\begin{figure*}[tp]
  \centering
  
  \begin{minipage}{0.48\textwidth}
    \centering
    \includegraphics[width=\linewidth]{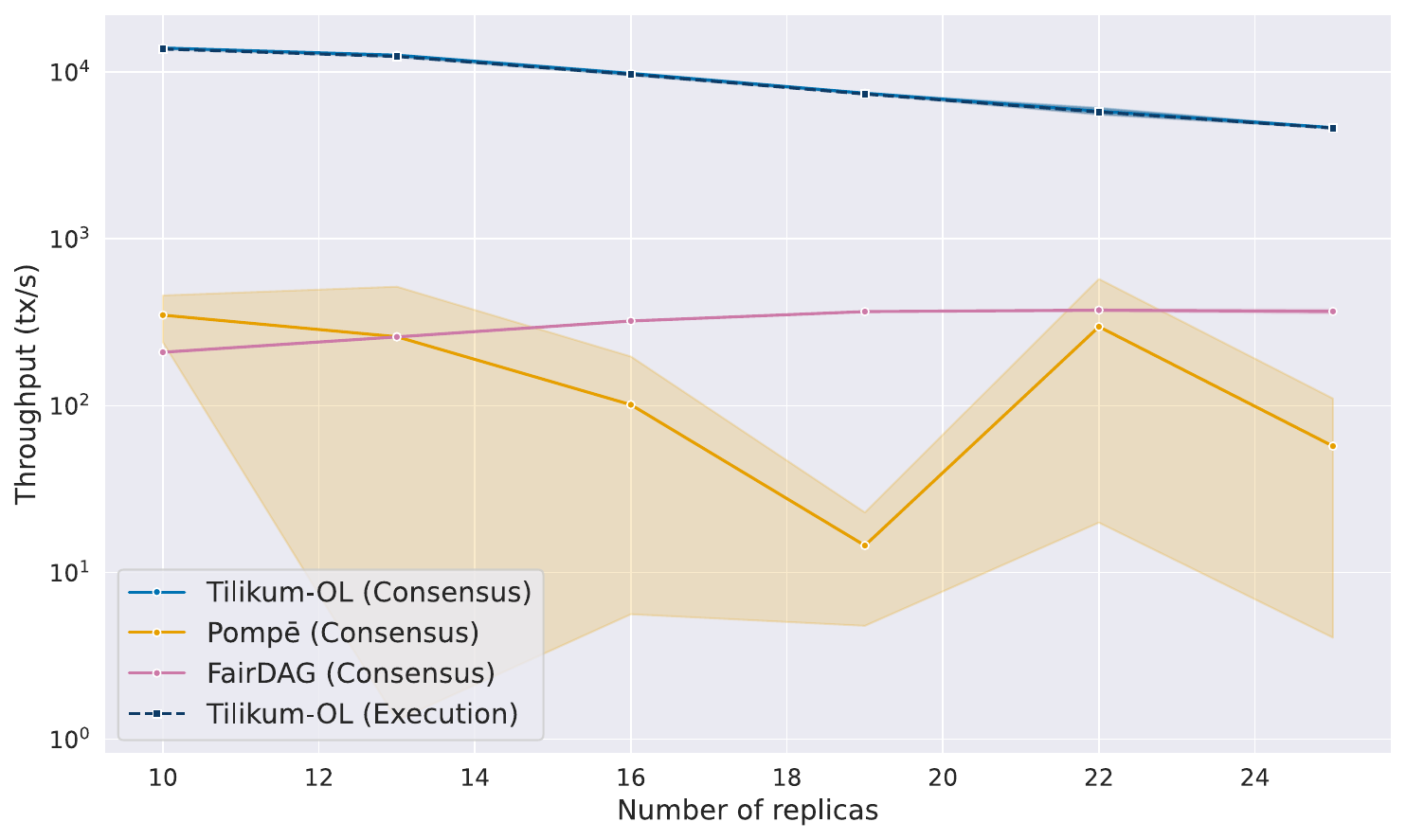}
    \caption{Ordering linearizability --- Execution throughput ($\uparrow$) depending on system size}
    \label{fig:throughput-n}
  \end{minipage}
  \hfill
  \begin{minipage}{0.48\textwidth}
    \centering
    \includegraphics[width=\linewidth]{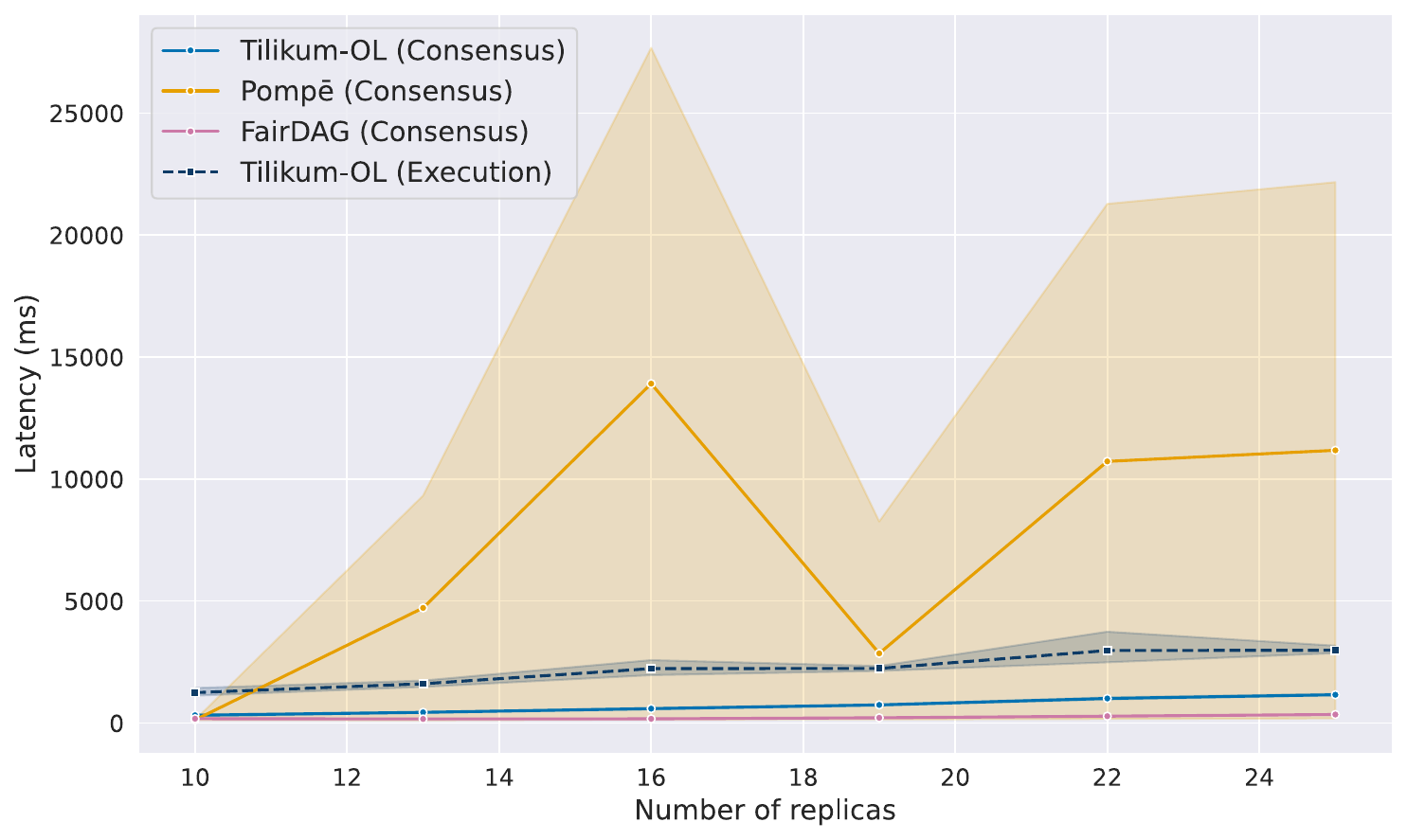}
    \caption{Ordering linearizability --- Execution latency ($\downarrow$) depending on system size}
    \label{fig:latency-n}
  \end{minipage}

  \vspace{2em} 

  \begin{minipage}{0.48\textwidth}
    \centering
    \includegraphics[width=\linewidth]{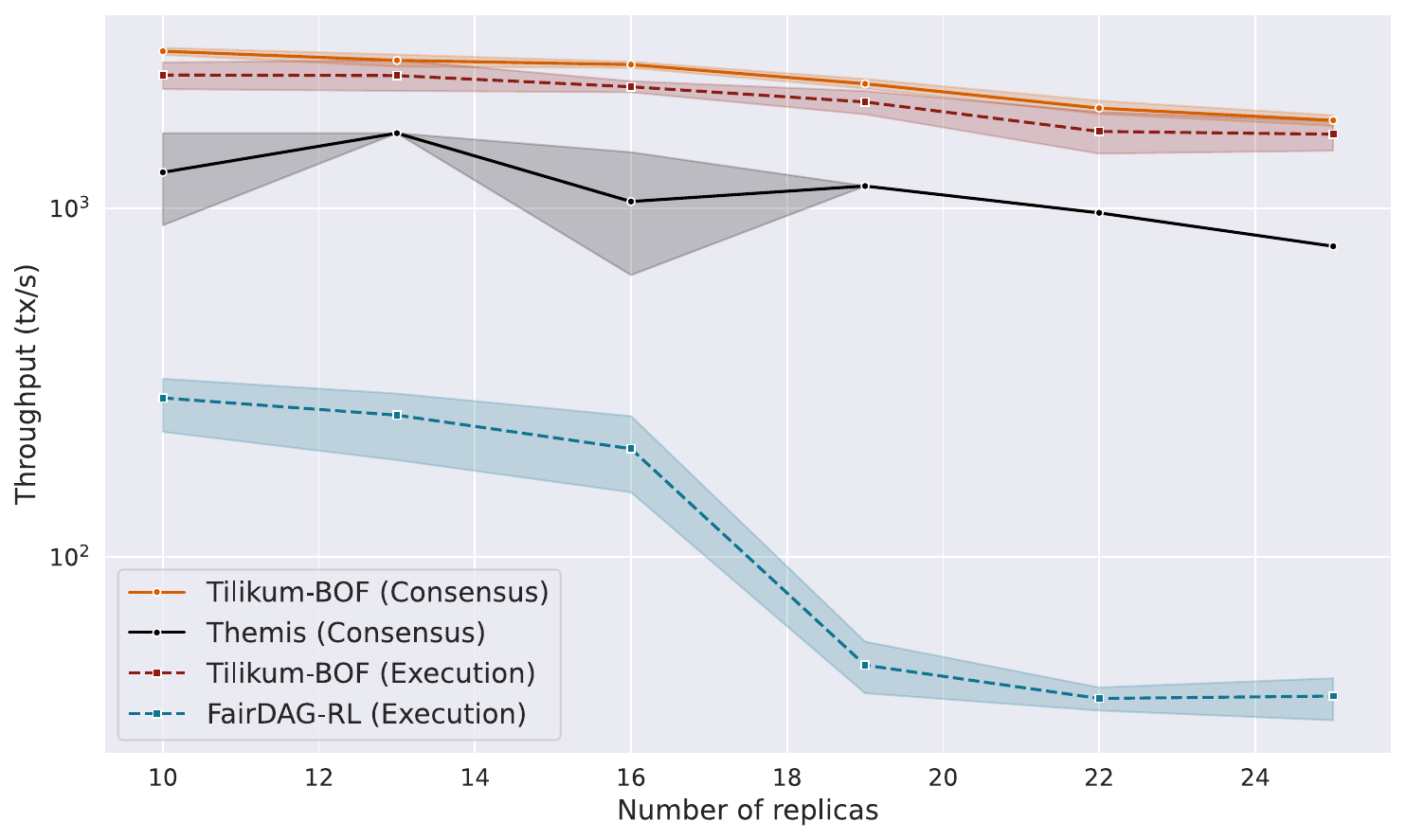}
    \caption{\gbof{} --- Execution throughput ($\uparrow$) depending on system size}
    \label{fig:throughput-tilikum-rl}
  \end{minipage}
  \hfill
  \begin{minipage}{0.48\textwidth}
    \centering
    \includegraphics[width=\linewidth]{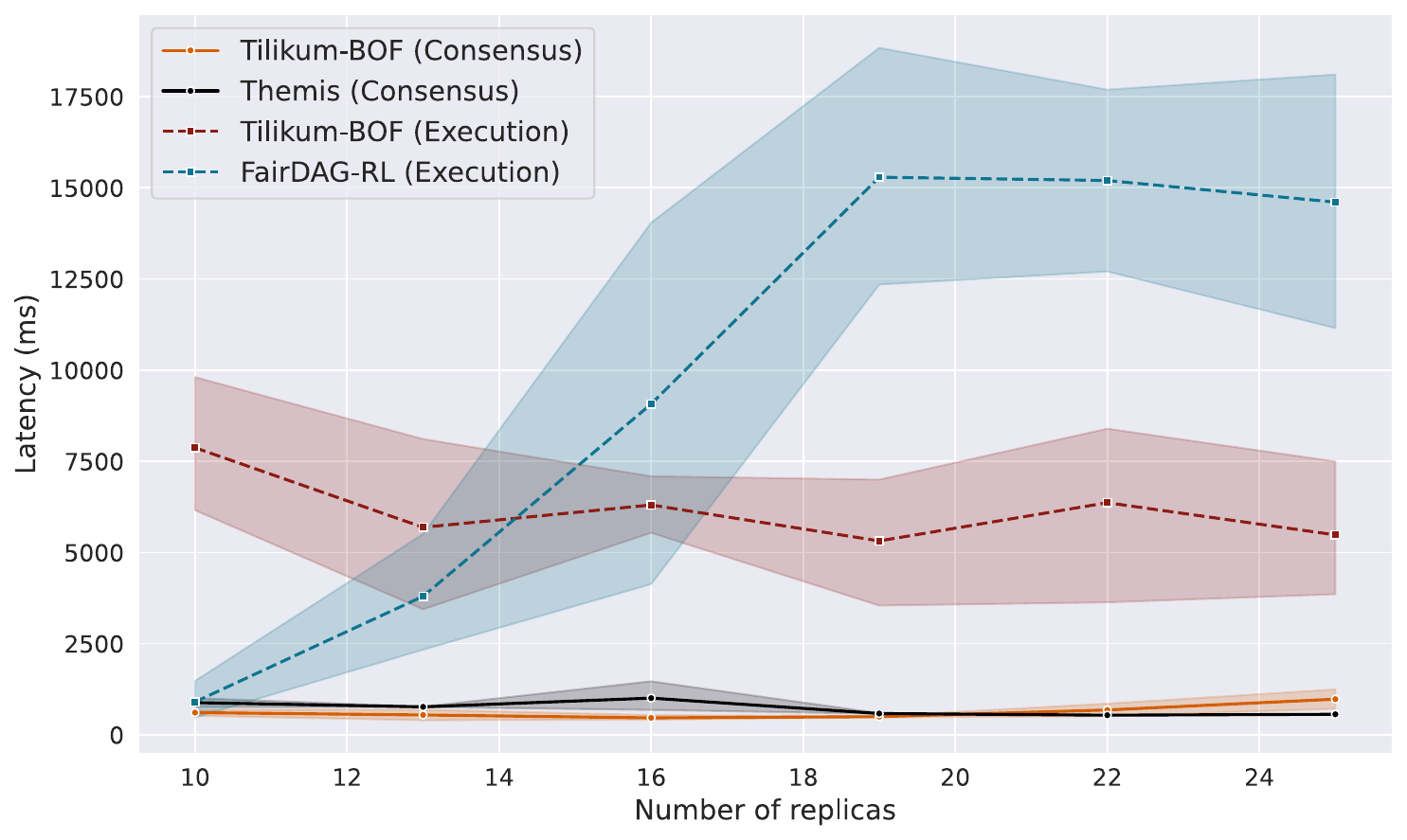}
    \caption{\gbof{} --- Execution latency ($\downarrow$) depending on system size}
    \label{fig:latency-tilikum-rl}
  \end{minipage}
\end{figure*}





\textbf{Consensus and Execution Throughput.}
Figs.~\ref{fig:throughput-n} and \ref{fig:throughput-tilikum-rl} report throughput when increasing the number of total parties.
We can notice that the overhead incurred by \name, as more data needs to be exchanged more times, is reduced on average by a factor of $4$ when compared to Narwhal/Tusk~\cite{danezis2022narwhal}.

\nameol remains faster than Pomp\=e~\cite{zhang2020byzantine}, with throughput of $14{,}000$\,tx/s at $N=10$, around $39$ times higher on average, and $4{,}600$\,tx/s at $N=25$, $81$ times higher.
Pomp\=e's runs vary by an order of magnitude in both throughput and latency (latency from $<1$\,s to $>10$\,s), likely due to its time-slotted consensus stalling on bad timestamp clusters.

FairDAG's consensus and execution throughput are always identical, showing how the combination of redundancy and weak edges allows transactions to be executed quickly.
This is also showcased by the low latency of the system.

The execution throughput of \nameol is, on average, $66$ times higher than
FairDAG-AB's at $N=10$ and $12$ times higher at $N=25$, although this might just be a consequence of the underlying consensus protocol, the fact that the Tusk~\cite{danezis2022narwhal} implementation in ResilientDB~\cite{gupta2020resilientdb} does not employ a worker-primary architecture, transaction redundancy and weak edges.
The results for \gbof show that \namebof achieves a consensus and execution throughput around $8$ times higher than FairDAG-RL at $N=10$ and $41$ times higher at $N=25$.
\namebof also exhibits throughput around $2$ times higher than Themis for all system sizes.
%
%
Silent parties have a very limited impact on the throughput of \nameol, as seen in Fig.~\ref{fig:throughput-f}.
This can be mostly explained by transactions now being proposed by multiple parties. \\

\textbf{Consensus and Execution Latency.} 
Figs.~\ref{fig:latency-n} and \ref{fig:latency-tilikum-rl} report latency when increasing the system size.
Latency is affected by the extra communication rounds required to collect timestamps, which is noticeable in the consensus latency.
Execution latency is also impacted, as the threshold needs to reach specific values to allow committed transactions to be executed.
The cost of managing the Logical Table is also not negligible.
We notice that in the case of $f=5$, shown in Fig.~\ref{fig:latency-f}, when all correct parties are able to successfully insert a block in every round there is a slight decrease in execution latency, moving it closer to the maximum theoretical value of the consensus latency.
FairDAG's~\cite{kang2026fairdag} usage of weak edges allows for the execution threshold to grow at a more consistent and reliable pace.
This is visible in the low and consistent latency values.

The time between receiving a transaction and it being included in a batch is slightly higher, as \name requires two rounds of reliable broadcast.
This results in latency being $1.2s$ at $N=10$ and $2.9s$ at $N=25$.
The value remains manageable, but for the smallest system size this is $9$ times slower than Pomp\=e and $2$ times slower than FairDAG-AB, while for the largest it's $5$ times slower than FairDAG-AB and $4$ times faster than Pomp\=e.

For \gbof, \namebof's latency is below $8s$ for $N=10$ and below $5.5s$ for $N=25$.
This positions it at around $10$ times slower than Themis.
\namebof is $9$ times slower than FairDAG-RL at $N=10$, but $3$ times faster for the largest system size.

\subsection{Reordering Attack Resistance.}

We evaluate the susceptibility of \name{}, and Narwhal/Tusk to the fissure, sluggish and speculative attacks proposed by Zhang et al.~\cite{zhang2025nofish}.
In these attacks, a malicious party attempts to front-run a transaction included within a correct party's block by generating a competing adversarial block. To achieve this, the attacker deviates from the standard consensus algorithm to ensure that their transaction is ordered before the victim transaction.
In the context of \nameol{}, where the execution order is determined by timestamps rather than by the order in which transactions appear in blocks, we consider a front-running attack to be successful if the attacker successfully commits an attacking block and if at least one transaction from its block is executed before a transaction from a victim block.
Despite adopting this broad success criterion, our results (Tab.~\ref{tab:asr_nodes_faults}) confirm that Zhang and Kate's attack are effective against the standard Narwhal/Tusk, with success rates ranging from 14\% to 95\% depending on the system configuration (system size and number of active malicious parties), but they also exhibit a 0\% success rate against \name{}.

\begin{table}[t]
\centering
\caption{Front-running success rate.}
\label{tab:asr_nodes_faults}
\resizebox{\columnwidth}{!}{
\begin{tabular}{p{2cm}cc|ccc}
\toprule
\textbf{Protocol} & \textbf{\# Parties} & \textbf{\# Faults} & \textbf{Fissure} & \textbf{Sluggish} & \textbf{Speculative} \\
\midrule
\multirow{3}{*}{Tusk} & 7 & 0 & 26.32\% & 36.59\% & 42.86\% \\
                              & 7 & 1 & 14.29\% & 46.15\% & 63.89\% \\
                              & 7 & 2 & 0.00\%  & 43.75\% & 93.55\% \\
\cmidrule{2-6}
\multirow{4}{*}{Tusk} & 10 & 0 & 27.81\% & 44.77\% & 75.46\% \\
                              & 10 & 1 & 36.11\% & 31.71\% & 66.67\% \\
                              & 10 & 2 & 50.00\% & 50.00\% & 57.58\% \\
                              & 10 & 3 & --      & 60.00\% & 95.24\% \\
\midrule
\textbf{\name}        & 10 -- 25 & 0 -- 5 & \multicolumn{3}{c}{\textbf{0.00\% (Resistant)}} \\
\bottomrule
\end{tabular}
}
\end{table}

\section{Related Work}
\label{sec:related_work}

\textbf{DAG-based protocols.}
Recent DAG-based consensus protocols improve block dissemination and reduce
latency. DAG-Rider~\cite{keidar2021all} builds a DAG via reliable broadcast and
uses a shared coin to order blocks. Narwhal/Tusk~\cite{danezis2022narwhal}
separate mempool and consensus, enabling garbage collection by removing weak edges and outperforming
DAG-Rider. Bullshark~\cite{spiegelman2022bullshark} adds a synchronous fast path,
while Mysticeti~\cite{babel2025mysticeti} leverages threshold clocks for pipelined
commits. Shoal~\cite{spiegelman2024shoal} and Shoal++~\cite{arun2025shoal++}
further reduce latency through leader reputation, optimized commit rules, and
parallel DAGs. Mahi-Mahi~\cite{jovanovic2025mahi} introduces an uncertified DAG
to reduce delay, and Starfish~\cite{polyanskii2025starfish} achieves linear
communication complexity via erasure coding. \\

\textbf{Order fairness.}
Fair-ordering properties aim to prevent adversarial reordering.
Pomp\=e~\cite{zhang2020byzantine} introduced ordering linearizability using
median timestamps. Aequitas~\cite{kelkar2020order} introduced $\gamma$-batch-order fairness, which Themis~\cite{kelkar2023themis}
and Lyra~\cite{zarbafian2023lyra} later extended. Wendy~\cite{wendy2020kursawe}
proposed probabilistic and timed fairness to guarantee termination. Other designs
include AOAB~\cite{gramoli2024aoab} (threshold signatures), Alpos et
al.~\cite{alpos2024eating} (anti-sandwiching), SpeedyFair~\cite{mu2024separation}
(decoupling ordering and consensus). Subsequent leader-based fairness work has explored data-dependent ordering~\cite{Nagda_UPenn_Rashnu_May2024}, weight-based sorting~\cite{Chen_China_Auncel_Aug2024}, position fairness~\cite{Wang_China_Dikaios_Oct2025}, cross-round graph maintenance~\cite{Ren_Melbourne_AUTIG_Oct2025}, bounded unfairness~\cite{Kiayias_Edinburgh_Taxis_2024}, and minimal batch variants~\cite{Ramseyer_Stanford_StreamingSocial_Oct2024}, but these target linear blockchains and are orthogonal to this work. 

FairDAG~\cite{kang2026fairdag} adds $\gamma$-batch-order fairness to DAG protocols that use weak edges, we compare extensively against FairDAG in this paper.
TEE-based approaches~\cite{ciampi2024universal} and $\kappa$-differential fairness schemes~\cite{cachin2022quick} have also been explored.
Amores-sesar et al.~\cite{rethinking2025amoressesar} proposed a timestamping mechanism to equip generic consensus protocols with the concept of time.
This approach uses local timestamps provided by parties, which is protocol independent, while byzantine fault tolerance is provided through mechanisms such as median for leader-based systems or quorum conditions in leaderless cases. \\

\textbf{Reordering attacks on DAG-based algorithms.}
Zhang et al.~\cite{zhang2025nofish} identified the fissure, sluggish, and speculative front-running attacks on Narwhal/Tusk.
In the fissure attack, attackers disconnect a victim's block from the rest of the DAG, forcing it to be ordered later by reducing its connectivity. In the speculative attack, attackers speculatively construct blocks with higher ordering priority to ensure their transactions precede the victim's.
Finally, in the sluggish attack, attackers deliberately create blocks in lower rounds that the ordering rule assigns higher priority, allowing them to be ordered ahead of newer victim blocks.
Mahe et al.~\cite{mahe2025order} showed a similar attack on DAG-Rider.
The Ambush attack~\cite{park2025frontrunning}
targets $\gamma$-batch-order fairness and has been demonstrated on single-leader protocols.
Our experiments show that \name{} perfectly defends against the fissure, sluggish and speculative attacks. 
Authors of the ambush attack propose a defensive strategy, i.e., immediate transaction dissemination upon receival, which is what Narwhal already does.
Evaluating the impact of the ambush attack on DAGs is future work as it requires significant computing power. \\

\textbf{Commit-reveal.}
Commit-reveal schemes~\cite{malkhi2022maximal,fernando2025trx} hide block contents
(e.g., via threshold encryption) to mitigate front-running. While compatible with
fair-ordering, they are orthogonal solutions and do not fully prevent reordering,
since attackers can still exploit metadata or sender information.

\section{Conclusion}

This paper presented \name{}, a DAG-based fair-ordering protocol that advances the state of the art by providing strong fairness guarantees without sacrificing scalability or resilience.
Unlike prior works, \name{} avoids reliance on weak edges, thereby enabling efficient garbage collection and reducing data redundancy, while maintaining robustness against malicious clients.
We proved that \name{} ensures safety and liveness.
We presented two versions of \name, namely \nameol and \namebof, that respectively guarantee the ordering linearizability and $\gamma$-batch-order fairness ordering properties.
Our experimental evaluation confirmed that \name{} matches or outperforms existing fair-ordering systems: \nameol{} sustains $14{,}000$\,tx/s at $N{=}10$ with $1.2$\,s latency ($39\times$ Pomp\=e's throughput) and stays $4\times$ faster than Pomp\=e at $N{=}25$, while \namebof{} doubles Themis' throughput at every evaluated system size. \name{} also fully prevents the reordering attacks of Zhang et al.~\cite{zhang2025nofish}, dropping their success rate from $14\%$--$95\%$ on vanilla Narwhal/Tusk to $0\%$.


\printbibliography


\newpage

\appendices

\section{Detailed Comparison with FairDAG}
\label{sec:details_FairDAG_tilikum}

In this section we discuss limitations of FairDAG~\cite{kang2026fairdag}, whether they could be mitigated and explain why \name does not share them.
The first part discusses the disadvantages of using weak edges and having redundant transactions in the DAG.
Finally, we propose three attack scenarios where a single malicious client can break liveness for the system, and how to prevent this.
Table~\ref{tab:FairDAG-compare} summarizes the main difference between \name and FairDAG.

\begin{table*}[t]
\small
  \centering
  \caption{Comparison between DAG-based fair ordering approaches}
  \label{tab:FairDAG-compare}
  \resizebox{\textwidth}{!}{
    \footnotesize
    \begin{tabular}{|p{3.5cm}|p{3cm}|p{3cm}|p{3cm}|}
      \hline
      & FairDAG~\cite{kang2026fairdag} & Narwhal/Tusk with full redundancy (FairDAG's approach without weak edges) & \name (this work) \\ \hline
      Mempool time complexity & NA & 2 & 4 \\ \hline
      Consensus time complexity & unmodified (full blocks of transactions) & unmodified (blocks of batch digests) & unmodified (blocks of batch digests) \\ \hline
      Garbage collection & \checkno & \checkyes & \checkyes \\ \hline
      Expected tx redundancy in DAG & $3f+1$ & $3f+1$ & $K \in [f+1, 3f+1]$ \\ \hline
      Expected \# signs per tx in DAG & $3f+1$ & $3f+1$ & $K (2f+1)$, or $K$ with BLS aggregate signatures (not implemented) \\ \hline
      Execution Liveness with $f$ faulty fast replicas & \checkyes & \checkno & \checkyes \\ \hline
      Robust against malicious clients & \checkno (\checkyes with perf.~degradation) & \checkno (\checkyes with perf.~degradation) & \checkyes \\ \hline
      Transaction Insertion Order & In-order & Out-of-order & Out-of-order \\ \hline
      Rebroadcasts & None & Full block of transactions & Hole Fillers + Batch digests \\ \hline
    \end{tabular}
  }
\end{table*}

\subsection{Weak Edges \& Redundancy}
\label{sec:redundancy}

FairDAG~\cite{kang2026fairdag} relies on weak edges for its implementation.
As a transaction needs to appear $2f{+}1$ times in the DAG to be executed, losing the strong theoretical guarantees given by the Validity property is not permissible.
As explained by Danezis et al.~\cite{danezis2022narwhal}, maintaining weak edges comes at the cost of losing garbage collection.
Without garbage collection, the system is unusable in practice, which we believe is a big design shortcoming.
The solution adopted by Narwhal/Tusk to achieve validity in the absence of weak edges is to retransmit the payload (i.e., transactions) of blocks that were not committed~\cite{danezis2022narwhal} together with the payload of the block for the round where this is detected.
This approach cannot be easily adapted to FairDAG.
The execution threshold relies on the assumption that the timestamps in the block created in round $r$ by a replica are always smaller than the ones in future rounds (i.e., $r+i$ for $i>0$).

If we allow the retransmission of the uncommitted payload of a block created in round $r$ during $r+i$, the blocks of rounds $r < r' < r+i$ will effectively include ordering indicators that are larger than the ones in the block for $r+i$, breaking the assumption.
Removing the need for weak edges is then not trivial and reduces the applicability of FairDAG in a realistic scenario.
To remove weak edges, replicas should rely on a way to detect the actual exact order of block's payloads and detect missing ones that might be retransmitted.
A mechanism similar to \name's logical clocks could be employed.

Apart from disallowing garbage collection, the way in which FairDAG computes the execution threshold and its reliance on weak edges \textbf{does not allow} for transactions to be inserted into the DAG in a different order than the one shown by their timestamps.
This assumption is relatively strong and can be problematic when paired with an asynchronous network.

\name does not require transactions to appear more than once in the DAG for it to be executed, as all necessary information is bundled in the batch's metadata.
As now the DAG will have redundant information, if we assume that no other factor is bottlenecking the system, its throughput should be at most the one achieved by the underlying DAG layer divided by the amount of (correct) parties.
This does not seem to be the case in the plots provided by the authors~\cite{kang2026fairdag}.
The author re-implemented Narwhal/Tusk~\cite{danezis2022narwhal} on the same framework they used for FairDAG~\cite{kang2026fairdag}, ResilientDB~\cite{gupta2020resilientdb}.
We think that the reason for this discrepancy is found in the way transactions are treated.

ResilientDB~\cite{gupta2020resilientdb} allows client to batch transactions.
These batches are then treated by the underlying protocol (in this case FairDAG on top of Narwhal/Tusk) as a single transaction.
This results in a batch of transactions being assigned a single timestamp, removing fairness between them.
Moreover, the throughput values were effectively multiplied by the batch size.
The default value available in the code was of $400$, which means that each $400$ transactions were assigned a single timestamp.

We also noticed that the script extracting the throughput and latency values from the execution logs was intentionally excluding values lower than $1,000$.
While we do not know the rationale behind this decision, we think it is not a correct behavior.
Both \name and FairDAG's execution speeds depend on a threshold which might not grow or grow particularly slowly in periods of network asynchrony or other types of faults.
We think that measuring the system's behavior under these circumstances is more correct, and we did so in our experiments.

\subsection{Attacks on Liveness with Malicious Clients}
\label{sec:-client-attack}

FairDAG assumes that ``clients generate and submit transactions to replicas, then await execution results in response'', and that ``there is no assumption that clients always behave correctly; they may exhibit arbitrary or malicious behavior''~\cite{kang2026fairdag}.
This however clashes with the later assumption that ``clients broadcast their transactions to all replicas''~\cite{kang2026fairdag}.

If we actually assume arbitrary client behavior, the system suffers from some of its design choices.
We now propose three scenarios in which malicious clients can halt the system and break liveness.
The first two scenarios are slight variations of each other, but they achieve different results.
The first one leads to an overall slowdown in throughput, while the second can potentially stop the execution threshold from advancing, completely blocking the system.
We then suggest two mitigation strategies which might have consequences on the scalability of the system.
The last attack strategy is more inherent to the system's assumptions.

The main way clients can misbehave is by not actually broadcasting their transactions, but instead send them to a subset of replicas.
This has negative outcomes for the client, as also pointed out by the authors, as malicious replicas could censor or delay the inclusion of transactions in their blocks.
However, the system makes no assumption that correct replicas will try to include transactions that have not been sent directly by a client.

A transaction received only by a subset of parties will then be included only by that same subset of replicas.
We will now see how different sizes of this subset impact the system, with two main types of attacks.

Let us start with the simple case where a client sends a transaction to $S < f{+}1$ parties.
In this case, those parties will use some of their block space to include this transaction, thus reducing throughput.
The LPAOI of this transaction will however always increase.
The entry in \texttt{lp\_ois} for the $S$ parties that received it will be constant, but for all other parties  it will keep increasing as the DAG expands.
Given that the LPAOI is the $(f{+}1)$th lowest value in the array and there are less than $f{+}1$ constant values, then it will always be changing.
The only consequence of this attack is reduced throughput, as block space is taken by this unexecutable transaction, and higher resource usage for the parties  that need to keep track of the Ordering Indicators for this transaction.

The second possibility for clients is more impactful.
If a malicious client sends a transaction $t$ to $f{+}1 < S < 2f{+}1$ parties, the LPAOI will be eventually computed to a constant value, but the transaction will never be executed.
As the rest of the DAG grows, the LPAOI of $t$ will eventually be the smallest and become the execution threshold for the system.
The transaction $t$ will never receive enough ordering indicators for execution and so it will always remain in the set of transactions without AOI, and always be the transaction determining the execution threshold.
The system will then fully come to a halt, with the DAG continuing to expand, but no new transaction being executed.

There are two ways in which this issue can be solved, and they both have some impact on the overall throughput of the system.
The simpler way requires correct replicas to re-broadcast all transactions they receive to all other replicas, similarly to the third mitigation proposed by Vafadar et al.~\cite{vafadar2023condorcet}.
The second attack scenario is fully stopped by this mitigation.
At least one correct replica will receive the transaction and broadcast it to all other replicas, making it eventually executable.
The first attack scenario is still possible, but its impact was not as severe as the second one to begin with.

If we want to fully stop the first attack as well, we can introduce a way for replicas to include transactions in their blocks even if they were not received by a client (or by another replica during the re-broadcasting step described above).
A correct replica that receives a valid block including a transaction digest $d$ that was not received by any client, can consider the receipt of this block as seeing the transaction for the first time, thus assigning its own ordering indicator and adding $d$ to its next block.
This is similar to what \name does.
A request for timestamps for a specific batch may include unseen transactions, which are then assigned a Timestamp Pair on the spot.

The last way malicious clients can influence the system has to do with how FairDAG stores transactions in the DAG.
Transactions are actually not stored in the ledger, but a cryptographically secure digest is used instead to reduce network and storage usage.
This works correctly assuming that all replicas have access to all transactions, and have a fast way to match the committed digest to the actual transaction content.
A malicious client partially sending a transaction, enough times to allow execution (i.e., $2f{+}1$), can stop parties  that did not receive it from knowing its contents, stopping their ability to maintain the correct state.

This attack is fully mitigated by the first strategy described above.
If at least one correct replica sees the transaction, they will broadcast it allowing every other correct party to store it.

Note that \name does not suffer from any of the above issues, all information required for executing a transaction is available in any block including the transaction and the execution threshold always grows, at least at the same speed as the slowest correct node.
We believe that one of the above mitigations should be implemented in FairDAG~\cite{kang2026fairdag}, so that their impact on the performance metrics can be measured.

\section{Additional Pseudocode}
\label{sec:pseudocode}

\begin{algorithm}
\small
\caption{First-Sight Timestamping at party $p_i$}
\label{alg:timestamp}
\begin{algorithmic}[1]
\State \textbf{Initialization:}
\State $\lc_i \gets 0$ \Comment{local logical counter}
\State $\mathit{known}_i \gets \emptyset$ \Comment{set of already-timestamped transactions}
\State $\mathit{store}_i \gets \{\}$ \Comment{map: $\tx \mapsto (\lc,\,\phiclock)$}

\Statex
\State \textbf{upon} receive transaction $\tx$ from client or replica \textbf{do}
  \If{$\tx \in \mathit{known}_i$}
    \State \textbf{return} \Comment{timestamp only on first sight}
  \EndIf
  \State $\lc_i(\tx) \gets \lc_i$ \Comment{assign current counter value}
  \State $\phiclock_i(\tx) \gets \proc{Clock}()$ \Comment{record physical time}
  \State $\lc_i \gets \lc_i + 1$
  \State $\mathit{known}_i \gets \mathit{known}_i \cup \{\tx\}$
  \State $\mathit{store}_i[\tx] \gets \bigl(\lc_i(\tx),\;\phiclock_i(\tx)\bigr)$
\end{algorithmic}
\end{algorithm}

For completeness, this appendix collects pseudocode for procedures of \name{} that were deferred from the main text: first-sight timestamping (Alg.~\ref{alg:timestamp}, Sec.~\ref{sec:parameter_k}), batch sealing  and DAG block creation (Alg.~\ref{alg:seal}, Sec.~\ref{sec:inclusion}), execution-threshold computation for \nameol{} (Alg.~\ref{alg:threshold}, Sec.~\ref{sec:details}) and per sub-DAG commit graph processing for \namebof{} (Alg.~\ref{alg:processsubdag}, Sec.~\ref{sec:graph-build}).
The notation below is shared by all algorithms. \\

\begin{algorithm}
\small
\caption{Batch Sealing (proposer side) and DAG block creation (round $r$) at party $p_i$ }
\label{alg:seal}
\begin{algorithmic}[1]
\Procedure{SealBatch}{$\mathit{batch}$}
  \Comment{$\mathit{batch}$: list of transactions to propose}
   \Statex
  \State{$\triangleright$ \textbf{Round 1:} collect $2f{+}1$ timestamp vectors}
  \Broadcast{$\langle \proc{TsRequest},\;\mathit{batch} \rangle$ to all parties}
  \State $\mathit{replies} \gets \{\}$
  \WaitUntil{$|\mathit{replies}| \ge 2f{+}1$}
  \State \hspace{\algorithmicindent} \textbf{upon} receive $\langle \proc{TsReply},\;\mathit{batch},\;\vec{v}_j,\;\sigma_j \rangle$ from $p_j$ \textbf{do}
      \If{$\proc{VerifySig}(\sigma_j,\, p_j)$}
        \State $\mathit{replies}[p_j] \gets \bigl[(\lc_j(\tx),\,\phiclock_j(\tx))\bigr]_{\tx \in \mathit{batch}}$
      \EndIf

  \Statex
  \State{$\triangleright$ \textbf{Round 2:} reliable storage of sealed batch}
  \State $\mathit{sealed} \gets (\mathit{batch},\;\mathit{replies})$
  \Broadcast{$\langle \proc{StoreReq},\;\mathit{sealed} \rangle$ to all parties}
  \State $\mathit{acks} \gets \{\}$
  \WaitUntil{$|\mathit{acks}| \ge 2f{+}1$}
  \State \hspace{\algorithmicindent} \textbf{upon} receive $\langle \proc{StoreAck},\; h,\; \sigma_j \rangle$ from $p_j$ \textbf{do}
      \If{$\proc{VerifySig}(\sigma_j,\, p_j)$}
        \State $\mathit{acks}[p_j] \gets \sigma_j$
      \EndIf
  \State \Return $h \gets \proc{Hash}(\mathit{sealed})$
    \Comment{ready for inclusion in DAG block}
\EndProcedure
\Statex
\Procedure{CreateBlock}{$r$}
  \State $\mathit{payloads} \gets \bigl\{\proc{Hash}(sb) \mid sb \in p_i.\mathit{sealedReady}\bigr\}$
  \State $\mathit{parents} \gets$ collect $2f{+}1$ certificates from round $r{-}1$
  \State $\mathit{fillers} \gets \Call{ComputeHoleFillers}{\,}$
    \Comment{Alg.~\ref{alg:hf}}
  \State $\mathit{block} \gets (p_i,\; r,\; \mathit{parents},\; \mathit{payloads},\; \mathit{fillers})$
  \State proceed with Narwhal consensus protocol
\EndProcedure
\end{algorithmic}
\end{algorithm}



\noindent\textbf{Notation.}
Each party $p_i$ maintains a monotonically increasing local counter $\lc_i$
and has access to a UNIX physical clock $\phiclock_i(\cdot)$.
A \emph{timestamp pair} for transaction $\tx$ at party $p_i$ is
$(\lc_i(\tx),\;\phiclock_i(\tx))$.
The set of $K$ parties responsible for including $\tx$ in a block is
determined by $\mathcal{H}(\tx)$ (a deterministic hash-based assignment).

\begin{algorithm}
\caption{Execution Threshold via Logical Table (\nameol{})}
\label{alg:threshold}
\begin{algorithmic}[1]
\Function{Threshold}{$L$}
  \State $H \gets [\;]$ \Comment{collect head timestamps}
  \ForAll{party $p_j$}
    \State $\mathit{head}_j \gets$ smallest $\lc$ in $L[p_j]$ where $\lc{+}1$ is absent
    \State $H.\proc{Append}(\mathit{head}_j.\phiclock)$
  \EndFor
  \State sort $H$ in ascending order
  \State \Return $\proc{Median}\bigl(H[1..2f{+}1]\bigr)$
\EndFunction
\end{algorithmic}
\end{algorithm}

\begin{algorithm}[t]
\small
\caption{Per-commit graph processing (\namebof{})}
\label{alg:processsubdag}
\begin{algorithmic}[1]
\State \textbf{State (long-lived):}
\State \hspace{\algorithmicindent} $\mathit{graphs} \gets [\,]$
\State \hspace{\algorithmicindent} $\mathit{node}[tx]$ with fields $\mathit{count}$, $\mathit{seqnums}$, $\mathit{graph}$
\State \hspace{\algorithmicindent} $F_p \gets 0$ for every party $p$
\Statex
\Procedure{ProcessSubdag}{$L_r,\, A_r$}
  \State $\mathcal{G}_r \gets$ \Call{NewGraph}{$r$}; \; $\mathit{graphs}.\proc{Push}(\mathcal{G}_r)$

  \Statex
  \State $\triangleright$ Advance frontiers and grow $\mathit{count}$
  \ForAll{$(p, \tx, \ell) \in A_r \cup \mathit{HoleFillers}(A_r)$}
    \State \Call{UpdateLvec}{$p,\tx,\ell$}
    \If{$F_p$ advanced past $\tx$ for the first time}
      \State $\mathit{node}[tx].\mathit{count} \mathrel{+}= 1$
    \EndIf
  \EndFor

  \Statex
  \State $\triangleright$ Admit transactions to $\mathcal{G}_r$
  \ForAll{$tx$ updated this round}
    \If{$\mathit{node}[tx].\mathit{graph} = \bot$ \textbf{and} $\mathit{node}[tx].\mathit{count} \ge n(1{-}\gamma)+f+1$}
      \State $\mathcal{G}_r.\mathit{nodes}.\proc{Add}(tx)$
      \State $\mathit{node}[tx].\mathit{graph} \gets \mathcal{G}_r$
    \EndIf
  \EndFor

  \Statex
  \State $\triangleright$ Aggregate pairwise weights and add edges
  \ForAll{newly admitted $tx \in \mathcal{G}_r$, $tx' \in \mathcal{G}_r.\mathit{nodes}$}
    \State update $\text{Weight}_\mathcal{L}(tx, tx')$ from new contributors
    \If{$\text{Weight} \ge n(1{-}\gamma)+f + 1$ \textbf{and} no edge exists}
      \State add edge in direction of larger weight
    \EndIf
  \EndFor

  \Statex
  \State $\triangleright$ Finalize tournaments in round order
  \ForAll{$\mathcal{G} \in \mathit{graphs}$ in round order}
    \If{$\mathcal{G}$ is not a tournament} \textbf{break} \EndIf
    \State $[S_1, \ldots, S_s] \gets$ \Call{TarjanSCC}{$\mathcal{G}$}
    \State $j^* \gets \max\{j \mid \exists\, tx \in S_j,\; \mathit{node}[tx].\mathit{count} \ge 2f{+}1\}$
    \For{$j = 1, \ldots, j^*$}
      \State append $S_j$ to final ordering
    \EndFor
    \For{$j = j^{*}{+}1, \ldots, s$}
      \State $\mathit{node}[tx].\mathit{graph} \gets \bot$ for $tx \in S_j$
    \EndFor
    \State release $\mathcal{G}$ from $\mathit{graphs}$
  \EndFor
\EndProcedure
\end{algorithmic}
\end{algorithm}

\end{document}